\documentclass[11pt,draftcls,onecolumn]{IEEEtran}
\usepackage[final]{graphicx}
\graphicspath{{../plots/}{../figures/}}
 \usepackage[pdf]{pstricks}
\usepackage{latexsym,bm} 
\usepackage{multirow}

\usepackage{verbatim}

\newcommand{\myexpnull}[1]{\mathsf{E}_{\infty}\left[#1\right]}
\newcommand{\probnull}[1]{\mathsf{P}_{\infty}\left\{#1\right\}}
\newcommand{\myall}{\mathsf{ALL}}
\newcommand{\mymax}{\mathsf{MAX}}
\newcommand{\myhall}{\mathsf{HALL}}

\newcommand{\DetRange}{r_d}
\newcommand{\InfRange}{r_i}
\newcommand{\ROI}{\mathsf{ROI}}
\newcommand{\mx}{\mathsf{MAX}}
\newcommand{\all}{\mathsf{ALL}}
\newcommand{\hall}{\mathsf{HALL}}

\newcommand{\CUSUM}{\mathsf{CUSUM}}
\newcommand{\add}{\mathsf{SADD}}

\newcommand{\tfa}{\mathsf{ARL2FA}}
\newcommand{\pfi}{\mathsf{PFI}}

\newcommand{\wsn}{\mathsf{WSN}}

\usepackage{amsmath}
\usepackage{subfigure}
\usepackage{amssymb}
\usepackage{theorem}
\newtheorem{theorem}{Theorem}

\newtheorem{proof}{Proof}
\newcommand{\pbmeasure}[1]{{\sf P}_1^{(i)}\left\{#1\right\}}
\newcommand{\pmeasure}[1]{{\sf P}_1^{({\bf d}(\ell_{e}))}\left\{#1\right\}}
\renewcommand{\leq}{\leqslant}
\renewcommand{\geq}{\geqslant}

\newcommand{\prob}[1]{\mathsf{P}\left\{#1\right\}}

\usepackage{url}

\def\smallskip{\vskip\smallskipamount}
\def\medskip{\vskip\medskipamount}
\def\bigskip{\vskip\bigskipamount}
\newcommand{\qed}{\hfill \rule{2.5mm}{2.5mm}}
{\theorembodyfont{\rmfamily}
\newtheorem{definition}{Definition}}
\newtheorem{lemma}{Lemma}
{\theorembodyfont{\rmfamily}
\newtheorem{example}{Example}} 

\newcommand{\nn}{\nonumber \\}

\begin{document}
\title{Distributed Detection/Isolation Procedures 
       for Quickest Event Detection in \\
	   Large Extent Wireless Sensor Networks}
%
%
%
\author{K.~Premkumar$^\S$,      
        Anurag~Kumar$^\dagger$, 
        and Joy~Kuri$^\ddagger$ 
\thanks{$\S$ K.~Premkumar's work in this paper was done during his 
        doctoral work at the Indian Institute of Science, Bangalore,
		India. He is currently with the Hamilton Institute, National
		University of Ireland, Maynooth, Ireland \newline (e--mail:
		Premkumar.Karumbu@nuim.ie).}
\thanks{$\dagger$ Anurag~Kumar is with the Department of Electrical
        Communication Engineering, Indian Institute of Science,
		Bangalore -- 560 012, India (e--mail:
		anurag@ece.iisc.ernet.in).}
\thanks{$\ddagger$ Joy Kuri is with the Centre for Electronics Design
        and Technology, Indian Institute of Science, Bangalore -- 560
		012, India (e--mail: kuri@cedt.iisc.ernet.in).}
\thanks{This is a revised and expanded version of a paper that was 
        presented in the 47th Annual Allerton Conference on 
		Communication, Control, and Computing, 2009. This work was 
		supported by a Project on Wireless Sensor Networks, funded by 
		DRDO, Government of India. The work of the second author was 
		also supported, in part, by the Department of Science and 
		Technology, through a J.C. Bose Fellowship.
		}
}
\maketitle

\begin{abstract}
We study a problem of distributed detection of a stationary point event
in a large extent wireless sensor network ($\wsn$), where the event
influences the observations of the sensors only in the vicinity of where
it occurs. An event occurs at an unknown time and at a random location
in the coverage region (or region of interest ($\ROI$)) of the $\wsn$.
We consider a general sensing model in which the effect of the event at
a sensor node depends on the distance between the event and the sensor
node; in particular, in the Boolean sensing model, all sensors in a disk
of a given radius around the event are equally affected. Following the
prior work reported in 
\cite{nikiforov95change_isolation},
\cite{nikiforov03lower-bound-for-det-isolation},
\cite{tartakovsky08multi-decision}, {\em the 
problem is formulated as that of detecting the event and locating it to
a subregion of the $\ROI$ as early as possible under the constraints
that the average run length to false alarm ($\tfa$) is bounded below by
$\gamma$, and the probability of false isolation ($\pfi$) is bounded
above by $\alpha$}, where $\gamma$ and $\alpha$ are target performance
requirements.  In this setting, we propose distributed procedures for
event detection and isolation (namely $\mx$, $\all$, and $\hall$), based
on the local fusion of $\CUSUM$s at the sensors. For these procedures,
we obtain bounds on the maximum mean detection/isolation delay ($\add$),
and on $\tfa$ and $\pfi$, and thus provide an upper bound on $\add$ as
$\min\{\gamma,1/\alpha\} \to \infty$.  For the Boolean sensing model,
we show that an asymptotic upper bound on the maximum mean
detection/isolation delay of our distributed procedure scales with
$\gamma$ and $\alpha$ in the same way as the asymptotically optimal
centralised procedure 
\cite{nikiforov03lower-bound-for-det-isolation}.
\end{abstract}

\vspace{-3mm}


\IEEEpeerreviewmaketitle

\section{Introduction}
\label{sec:introduction}
Event detection is an important application for which a wireless sensor
network ($\wsn$) is deployed. A number of sensor nodes (or ``motes'')
that can sense, compute, and communicate are deployed in a region of
interest ($\ROI$) in which the occurrence of an event (e.g., crack in a
structure) has to be detected. In our work, we view {\em an event as
being associated with a change in the distribution (or cumulative
distribution function) of a physical quantity that is sensed by the
sensor nodes}. Thus, the work we present in this paper is in the
framework of quickest detection of change in a random process. In the
case of small extent networks, where the coverage of every sensor spans
the whole $\ROI$, and where we assume that an event affects all the
sensor nodes in a statistically equivalent manner, we obtain the
classical change detection problem whose solution is well known 
(see, for example, \cite{shiryayev63},
\cite{stat-sig-proc.page54continuous-inspection-schemes},
\cite{stat-sig-proc.tartakovsky-veeravalli03quickest-change-detection}).
In \cite{premkumar_etal10det_over_mac} and
\cite{premkumar_kumar08infocom}, we have studied variations of the
classical problem in the $\wsn$ context, where there is a wireless
communication network between the sensors and the fusion
centre~\cite{premkumar_etal10det_over_mac}, and where there is a cost
for taking sensor measurements~\cite{premkumar_kumar08infocom}.

However, in the case of large extent networks, where the $\ROI$ is large
compared to the coverage region of a sensor, an event (e.g., a crack in
a huge structure, gas leakage from a joint in a storage tank) affects
sensors that are in its proximity; further the effect depends on the
distances of the sensor nodes from the event. Since the location of the
event is unknown, {\em the post--change distribution of the observations
of the sensor nodes are not known}. In this paper, we are interested in
obtaining procedures for detecting and locating an event in a large
extent network. This problem is also referred to as {\em change
detection and isolation} 
(see \cite{nikiforov95change_isolation},
\cite{nikiforov03lower-bound-for-det-isolation},
\cite{tartakovsky08multi-decision}, 
\cite{stat-sig-proc.mei05information-bounds},
\cite{lai00multi-hypothesis-testing}). 
Since the $\ROI$ is large, a large number of sensors are deployed to
cover the $\ROI$, making a centralised solution infeasible. In our work,
{\em we seek distributed algorithms for detecting and locating an event,
with small detection delay, subject to constraints on false alarm and
false isolation}. The distributed algorithms require only local
information from the neighborhood of each node.

\subsection{Discussion of Related Literature}
The problem of sequential change detection/isolation with a finite set
of post--change hypotheses was introduced by Nikiforov
\cite{nikiforov95change_isolation}, where he 
studied the  change
detection/isolation problem with the observations being conditionally
independent, and proposed a non--Bayesian procedure which is shown to be
maximum mean detection/isolation delay optimal, as the average run
lengths to false alarm and false isolation go to $\infty$. Lai
\cite{lai00multi-hypothesis-testing} considered the multi--hypothesis
change detection/isolation problem with stationary pre--change and
post--change observations, and obtained asymptotic lower bounds for the
maximum mean detection/isolation delay.

Nikiforov also studied a change detection/isolation problem under the
average run length to false alarm ($\tfa$) and the probability of false isolation
($\pfi$) constraints \cite{nikiforov03lower-bound-for-det-isolation}, in which he
showed that a ${\sf CUSUM}$--like {\em recursive} procedure is asymptotically
maximum mean detection/isolation delay optimal among the procedures that
satisfy $\tfa\geq\gamma$ and $\pfi\leq\alpha$ asymptotically, as  
$\min\{\gamma,1/\alpha\}\to\infty$ . Tartakovsky in
\cite{tartakovsky08multi-decision} also studied the change
detection/isolation problem where he proposed recursive matrix ${\sf
CUSUM}$ and recursive matrix Shiryayev--Roberts tests, and showed that
they are asymptotically maximum mean delay optimal 
over the constraints $\tfa\geq\gamma$ and $\pfi\leq\alpha$ asymptotically, as  
$\min\{\gamma,1/\alpha\}\to\infty$.

Malladi and Speyer \cite{malladi_speyer99shiryayev_isolation} 
studied a Bayesian change detection/isolation problem and obtained a
mean delay optimal centralised procedure which is a threshold based rule
on the a posteriori probability of change corresponding to each
post--change hypothesis.	

Centralised procedures incur high communication costs and distributed
procedures would be desirable. In this paper, we study distributed
procedures based on $\CUSUM$ detectors at the sensor nodes where the
$\CUSUM$ detector at sensor node $s$ is driven only by the observations
made at node $s$. Also, in the case of large extent networks, the
post--change distribution of the observations of a sensor node, in
general, depends on the distance between the event and the sensor node
which is unknown.

\subsection{Summary of Contributions}
\begin{enumerate}
\item As the ${\sf WSN}$ considered is of large extent, the post--change
	  distribution is unknown, and could belong to a set of
	  alternate hypotheses.  
	  In Section~\ref{sec:problem_formulation}, we formulate the event
      detection/isolation problem in a large extent network in the
	  framework of 
      \cite{nikiforov03lower-bound-for-det-isolation},
      \cite{tartakovsky08multi-decision}
	  as a maximum mean detection/isolation delay minimisation problem subject to an
	  average run length to false alarm ($\tfa$) and probability of
	  false isolation ($\pfi$) constraints. 
	  
\item We propose distributed detection/isolation procedures ${\sf
      MAX}$, ${\sf ALL}$, and ${\sf HALL}$ ({\bf H}ysteresis modified 
	  {\bf ALL}) for large extent networks in
	  Section~\ref{sec:distributed_change_detection_isolation_procedures}.
	  The procedures 
	  ${\sf MAX}$ and ${\sf ALL}$ are extensions of the
	  decentralised procedures $\sf{MAX}$
	  \cite{stat-sig-proc.tartakovsky-veeravalli03quickest-change-detection}
	  and $\sf{ALL}$ \cite{stat-sig-proc.mei05information-bounds},
	  \cite{agt-vvv08}, which were developed for small extent networks.
	  The distributed procedures 
	  are energy--efficient compared to the centralised procedures.
	  Also, the known centralised procedures are applicable only for the
	  Boolean sensing model. 
	  
\item In 
      Section~\ref{sec:distributed_change_detection_isolation_procedures},
	  we first obtain bounds on $\tfa$, $\pfi$, and maximum mean 
	  detection/isolation delay ($\add$) for the distributed procedures
	  $\sf{MAX}$, $\sf{ALL}$, and $\sf{HALL}$. These bounds are then
	  applied to get an upper bound on the $\add$ for the procedures
	  when $\tfa \geq \gamma$, and $\pfi \leq \alpha$, where $\gamma$
	  and $\alpha$ are some performance requirements. For the case of
	  the Boolean sensing model, we compare the $\add$ 
	  of the distributed procedures with
	  that of Nikiforov's procedure
      \cite{nikiforov03lower-bound-for-det-isolation}
	  (a centralised asymptotically optimal procedure) 
	  and show that the
an asymptotic upper bound on the maximum mean
detection/isolation delay of our distributed procedure scales with
$\gamma$ and $\alpha$ in the same way as that of  
\cite{nikiforov03lower-bound-for-det-isolation}.
\end{enumerate}

\section{System Model}
\label{sec:system_model}
Let $\mathcal{A} \subset \mathbb{R}^2$ be the region of interest
($\ROI$) in which $n$ sensor nodes are deployed. All nodes are equipped
with the same type of sensor (e.g., acoustic). 
Let
$\ell^{(s)}\in\mathcal{A}$ be the location of sensor node $s$, and
define ${\bm \ell} :=[\ell^{(1)},\ell^{(2)},\cdots,\ell^{(n)}]$. We
consider a discrete--time system, with the basic unit of time being one
slot, indexed by $k=0,1,2,\cdots,$ the slot $k$ being the time interval
$[k,k+1)$. The sensor nodes are assumed to be time--synchronised (see,
for example, \cite{solis-etal06time-synch}), and at the beginning of
every slot $k \geqslant 1$, each sensor node $s$ samples its environment
and obtains the observation $X_k^{(s)}\in\mathbb{R}$.

\subsection{Change/Event Model}
An event (or change) occurs at an unknown time $T \in \{1,2,\cdots\}$
and at an unknown location $\ell_e \in \mathcal{A}$. We consider only
stationary (and permanent or persistent) point events, i.e., an event
occurs at a point in the region of interest, and {\em having occurred,
stays there forever}. Examples that would motivate such a model are 1)
gas leakage in the wall of a large storage tank, 2) excessive strain at
a point in a large 2--dimensional structure. In
\cite{tartakovsky-report} and \cite{premkumar_etal10iwap}, the authors
study change detection problems in which the event stays only for a
finite random amount of time.

An event is viewed as a source of some physical signal that can be
sensed by the sensor nodes.  Let $h_e$ be the signal strength of the
event\footnote{In case, the signal strength of the event is not known,
but is known to lie in an interval $[\underline{h}, \overline{h}]$, we
work with $h_e = \underline{h}$ as this corresponds to the least
Kullback--Leibler divergence between the ``{\em event not occurred}''
hypothesis and the ``{\em event occurred}'' hypothesis. See
\cite{stat-sig-proc.tartakovsky-polunchenko08change-det-with-unknown-param}
for change detection with unknown parameters for a collocated network.}. 
A sensor at a distance $d$ from the event senses a signal  $h_e \rho(d)
+ W$, where $W$ is a random zero mean noise, and $\rho(d)$ is
the distance dependent loss in signal strength which is a
decreasing function of the distance $d$, with $\rho(0)=1$. We
assume an isotropic distance dependent loss model, whereby the
signal received by all sensors at a distance $d$ (from the event)
is the same.

\begin{example} {\bf The Boolean model} (see \cite{sensing-models}): In
this model, the signal strength that a sensor receives is the same
(which is given by $h_e$) when the event occurs within a distance of
$\DetRange$ from the sensor and is 0 otherwise. Thus, for a Boolean 
sensing model,
\begin{eqnarray*}  
\rho(d) & = & \left\{
\begin{array}{ll}
1, & \text{if} \ d \leqslant \DetRange\\
0, & \text{otherwise}. 
\end{array}
\right.
\end{eqnarray*}
\end{example}
\begin{example} {\bf The power law path--loss model} (see
\cite{sensing-models}) is given by
\begin{eqnarray*}  
\rho(d) & = & d^{-\eta}, 
\end{eqnarray*} 
for some path loss exponent $\eta > 0$. For free space, $\eta = 2$.
\end{example}

\subsection{Detection Region and Detection Partition}
\label{sec:detection-range}
In Example 2, we see that the signal from an event varies continuously
over the region. Hence, unlike the Boolean model, there is no clear
demarcation between the sensors that observe the event and those that do
not. Thus, in order to facilitate the design of a distributed detection
scheme with some performance guarantees, in the remainder of this
section, we will define certain regions around each sensor.
\begin{definition}
Given $0 < \mu_1 \leqslant h_e$, the {\bf Detection Range} $\DetRange$
of a sensor is defined as
the distance from the sensor within which the occurrence of an event
induces a signal level of at least $\mu_1$, i.e.,
\begin{align*}
\DetRange &:=  \sup\left\{d : h_e \rho(d) \ge \mu_1\right\}. 
\end{align*}
\hfill\qed
\end{definition}

\vspace{-4mm}

In the above definition, $\mu_1$ is a design parameter that defines
the acceptable detection delay. For a given signal strength $h_e$, a
large value of $\mu_1$ results in a small detection range $\DetRange$ (as
$\rho(d)$ is non--increasing in $d$). We will see in
Section~\ref{sec:average_detection_delay}
(Eqn.~\eqref{eqn:sadd_max_all_hall}) that the $\add$
of the distributed change detection/isolation procedures we propose,
depends on the detection range $\DetRange$, and that a small $\DetRange$ (i.e., a
large $\mu_1$) results in a small $\add$, while
requiring more sensors to be deployed in order to achieve coverage of
the ${\sf ROI}$.

\begin{figure}[t]
\centering
\includegraphics[width = 55mm, height = 50mm]{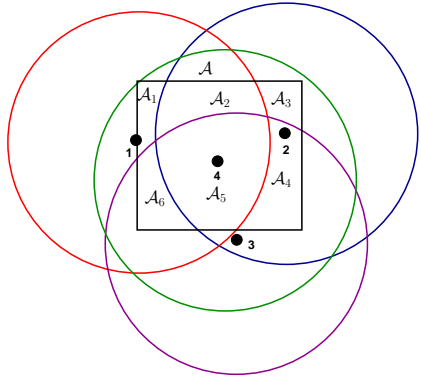}
\caption{{\bf Partitioning of $\mathcal{A}$ in a large $\wsn$ by
         detection regions}: (a simple example) The coloured solid 
		 circles
		 around each sensor node denote their detection regions. The
		 four sensor nodes divide the $\ROI$, indicated
		 by the square region, into regions $\mathcal{A}_1, \cdots,
		 \mathcal{A}_6$ such that region $\mathcal{A}_i$ is
		 detection--covered by a unique set of sensors $\mathcal{N}_i$.
		 For example, ${\cal A}_1$ is detection covered by the set of
		 sensors ${\cal N}_1 = \{1,2,4\}$, etc.	 
         }
\label{fig:coverage}
\end{figure}
We say that a location $x \in \ROI$ is {\em detection--covered} by
sensor node $s$, if $\|\ell^{(s)}-x\| \leqslant \DetRange$.  For any
sensor node $s$, $\mathcal{D}^{(s)} := \{x \in \mathcal{A} :
\|\ell^{(s)}-x\| \leqslant \DetRange \}$ is called its {\em
detection--coverage region} (see Fig.~\ref{fig:coverage}).  {\em We
assume that the sensor deployment is such that every $x \in \mathcal{A}$
is detection--covered by at least one sensor} (Fig.~\ref{fig:coverage}). For each $x \in \mathcal{A}$, define
$\mathcal{N}(x)$ to be the largest set of sensors by which $x$ is
detection--covered, i.e., $\mathcal{N}(x) := \{s : x \in {\cal
D}^{(s)}\}$. Let $\mathcal{C}(\mathcal{N}) = \{\mathcal{N}(x) : x \in
{\cal A} \}$. $\mathcal{C}(\mathcal{N})$ is a finite set and
can have at most $2^n-1$ elements.  Let $N =
|\mathcal{C}(\mathcal{N})|$. For each $\mathcal{N}_i \in
\mathcal{C}(\mathcal{N})$, we denote the corresponding
detection--covered region by $\mathcal{A}_i = \mathcal{A}(\mathcal{N}_i)
:= \{x \in \ROI : \mathcal{N}(x) = \mathcal{N}_i \}$.  Evidently, the
${\cal A}_i, 1 \leqslant i \leqslant N$, partition the $\ROI$.  We say
that the $\ROI$ is {\em detection--partitioned} into a {\em minimum
number of subregions}, $\mathcal{A}_1, \mathcal{A}_2, \cdots,
\mathcal{A}_N$, such that the subregion $\mathcal{A}_i$ is
detection--covered by a unique set of sensors $\mathcal{N}_i$, and
$\mathcal{A}_i$ is the maximal detection--covered region of
$\mathcal{N}_i$, i.e., $\forall i \neq i'$, $\mathcal{N}_i \neq
\mathcal{N}_{i'}$ and $\mathcal{A}_i\cap \mathcal{A}_{i'}=\emptyset$.
See Fig.~\ref{fig:coverage} for an example.  

\subsection{Sensor Measurement Model}
\label{subsec:measurement_model}
Before change, i.e., for $k < T$, the observation $X_k^{(s)}$ at the
sensor $s$ is just the zero mean sensor noise $W_k^{(s)}$, the
probability density function (pdf) of which is denoted by $f_0(\cdot)$
({\em pre--change pdf}). After change, i.e., for $k \geqslant T$ with the
location of the event being $\ell_e$, the observation of sensor $s$ is
given by $X_k^{(s)} = h_e\rho(d_{e,s}) + W_k^{(s)}$ where $d_{e,s} :=
\|\ell^{(s)} - \ell_e\|$, the pdf of which is denoted by
$f_1(\cdot;d_{e,s})$ ({\em post--change pdf}). The noise processes
$\{W_k^{(s)}\}$ are independent and identically distributed (iid) across
time and across sensor nodes. In the rest of the paper, {\em we consider
$f_0(\cdot)$ to be Gaussian with mean 0 and variance
$\sigma^2$.}

We denote the probability measure when the change happens at time $T$
and at location $\ell_e$ by ${\sf P}^{({\bf
d}(\ell_e))}_{T}\left\{\cdot\right\}$, where ${\bf d}(\ell_e) =
[d_{e,1},d_{e,2},\cdots,d_{e,n}]$, and the corresponding expectation
operator by ${\sf E}^{({\bf d}(\ell_e))}_{T}\left[\cdot\right]$.  In the
case of Boolean sensing model, the post--change pdfs depend only on the
detection subregion where the event occurs, and hence, we denote the
probability measure when the event occurs at $\ell_e \in {\cal A}_i$ and
at time $T$ by ${\sf P}^{(i)}_{T}\left\{\cdot\right\}$, and the
corresponding expectation operator by ${\sf
E}^{(i)}_{T}\left[\cdot\right]$.


\subsection{Local Change Detectors}
We compute a $\CUSUM$ statistic $C_k^{(s)}, k\geq 1$ at each sensor $s$
based only on its own observations. 
The $\CUSUM$ procedure was proposed by Page 
{\cite{stat-sig-proc.page54continuous-inspection-schemes} as a solution
to the classical change detection problem (${\sf CDP}$, in which there
is one pre--change hypothesis and only one post--change hypothesis).
The optimality of $\CUSUM$ was shown for conditionally iid observations
by Moustakides in \cite{moustakides86optimal-stopping-times} for a 
maximum mean delay metric introduced by Pollak \cite{pollak85} which is 
$\mathsf{SADD}(\tau) :=$ $\underset{T \geqslant 1}{\sup} \ \ 
{\mathsf E}_T\left[\tau-T|\tau \geq T\right]$.

The driving term of $\CUSUM$ should
be the log likelihood--ratio (LLR) of $X_k^{(s)}$ defined as
$Z^{(s)}_k(d_{e,s}) :=
\ln\left(\frac{f_1(X_k^{(s)};d_{e,s})}{f_0(X_k^{(s)})}\right)$. As the
location of the event $\ell_e$ is unknown, the distance $d_{e,s}$ is
also unknown. Hence, one cannot work with the pdfs $f_1(\cdot;d_{e,s})$.
We propose to drive the $\CUSUM$ at each node $s$ with
$Z^{(s)}_k(\DetRange)$, where we recall that $\DetRange$ is the
detection range of a sensor.  Based on the $\CUSUM$ statistic
$C_k^{(s)}, k\geq 1$, sensor $s$ computes a sequence of local decisions
$D_k^{(s)} \in \{0,1\}, k \geq 1$, where 0 represents no--change and 1
represents change. For each set of sensor nodes ${\cal N}_i$ that
detection partitions the $\ROI$, we define $\tau^{({\cal N}_i)}$, the
stopping time (based on the sequence of local decisions $D_k^{(s)}$s for
all $s \in {\cal N}_i$) at which the set of sensors ${\cal N}_i$ detects
the event. The way we obtain the local decisions $D_k^{(s)}$ from the
$\CUSUM$ statistic $C_k^{(s)}, k \geq 1$, and the way these local
decisions determine the stopping times $\tau^{({\cal N}_i)}$, varies 
from rule to rule. Specific rules for local decision and the fusion of 
local decisions will be described in
Section~\ref{sec:distributed_change_detection_isolation_procedures} 
(also see \cite{mandar-etalfusum}). 

An implementation strategy for our distributed event detection/isolation
procedure can be the following. We assume that the sensors know to
which detection sensor sets ${\cal N}_i$s they belong. This could be
done by initial configuration or by self--organisation. When the local
decision of sensor $s$ is 1, it broadcasts this fact to
all sensors in its detection neighbourhood.  In practise, the broadcast
range of these radios is substantially larger than the detection range.
Hence, the local decision of $s$ is learnt by all sensors $s'$ that
belong to ${\cal N}_i$ to which $s$ belongs. When any node learns that
all the sensors in ${\cal N}_i$ have reached the local decision 1, it
transmits an alarm message to the base station \cite{thuli09thesis}. A
distributed leader election algorithm can be implemented so that only
one, or a controlled number of alarms is sent. 
This alarm message is
carried by geographical forwarding \cite{naveen-kumar10geo-forwarding}.
A system that utilises such local fusion (but with a different sensing
and detection model) was developed by us and is reported in
\cite{wsn10smart-detect}.

\subsection{Influence Region}
\label{subsec:influence_region}
After a set of nodes ${\cal N}_i$ declares an event, the
event is {\em isolated} to a region associated with ${\cal N}_i$ 
called the influence region. 
In the Boolean sensing model, if an event occurs in ${\cal A}_i$, then
only the sensors $s \in {\cal N}_i$ observe the event, while the other
sensors ${s' \notin {\cal N}_i}$ only observe noise. On the other hand,
in the power law path--loss model, sensors ${s' \notin {\cal
N}_i}$ can also observe the event, and the driving term of the $\CUSUM$s of sensors
$s'$ may be affected by the event. The mean of the driving term of
$\CUSUM$ of any sensor $s$ is given by
\begin{eqnarray}\label{eqn:proof_of_lemma_1}
{\mathsf E}_{f_1(\cdot;d_{e,s})} [Z_k^{(s)}(\DetRange)] 
&  = & \frac{(h_e\rho(\DetRange))^2}{2\sigma^2}
	   \left(\frac{2\rho(d_{e,s})}{\rho(\DetRange)} - 1\right).
\end{eqnarray}
Thus, the mean of the increment that drives $\CUSUM$ of node
$s$ decreases with $d_{e,s}$ and becomes negative when $2\rho(d_{e,s}) <
\rho(\DetRange)$. In this region, we are interested in finding $T_E$, the
expected time for the $\CUSUM$ statistic $C_k^{(s)}$ to cross the
threshold $c$. Define $\tau^{(s)} := \inf \left\{k: C_k^{(s)} \geqslant
c\right\}$, and hence, $T_E =  {\sf E}_{1}^{({\bf
d}(\ell_e))}\left[\tau^{(s)}\right]$.

\begin{lemma}
\label{lem:mean-time-tfi}
If the distance between sensor node $s$ and the event, $d_{e,s}$ is such
that $2\rho(d_{e,s}) < \rho(\DetRange)$, then  
\begin{eqnarray*}
T_E & \geqslant &\exp(\omega_0 c) 
\end{eqnarray*}
where 
$\omega_0 = 1 - \frac{2\rho(d)}{\rho(\DetRange)}$. 
\end{lemma}
\begin{proof}
From (Eqn.~5.2.79 pg. 177 of) \cite{basseville-nikiforov93detection}, we
can show that ${\sf E}_{1}^{({\bf d}(\ell_e))}\left[\tau^{(s)}\right]
\geqslant \exp(\omega_0 c)$ where $\omega_0$ is the solution to the
equation
\[
{\sf E}_{1}^{({\bf d}(\ell_e))}\left[ e^{\omega_0 Z_k^{(i)}(\DetRange)}\right] = 0,
\] 
which is given by $\omega_0 = 1 - \frac{2\rho(d)}{\rho(\DetRange)}$ 
(see Eqn.~\eqref{eqn:proof_of_lemma_1}). 
\end{proof}
We would be interested in ${\sf T_E} \geqslant
\exp(\underline{\omega}_0\cdot c)$ for some $0 < \underline{\omega}_0 < 1$.
We now define the {\em influence
range} of a sensor as follows. 

\begin{definition}{\bf Influence Range} of a sensor, $\InfRange$, is 
defined as the distance from the sensor within which the occurrence of 
an event can be detected within a mean delay of 
$\exp{(\underline{\omega}_0 c)}$ where $\underline{\omega}_0$ 
is a parameter of interest and $c$ is the threshold of the local $\CUSUM$ detector.
Using Lemma~\ref{lem:mean-time-tfi}, we see that 
$\InfRange =\min\{d' : 2\rho(d') \leqslant (1 - \underline{\omega}_0)
\rho(\DetRange)\}$. 
\qed
\end{definition}
A location $x \in {\cal A}$ is influence covered by a sensor
$s$ if $\|\ell^{(s)}-x\| \leq \InfRange$, and a set of sensors ${\cal
N}_j$ is said to influence cover $x$ if each sensor $s \in {\cal N}_j$
influence covers $x$.

From Lemma~\ref{lem:mean-time-tfi}, we see that by having a large value
of $\underline{\omega}_0$, i.e., $\underline{\omega}_0$ close to 1, the
sensors that are beyond a distance of $\InfRange$ from the event take a
long time to cross the threshold.  However, we see from the definition
of influence range that a large value of $\underline{\omega}_0$ gives a
large influence range $\InfRange$. We will see from the discussion in
Section~\ref{subsec:discussion} that a large influence range results in
the isolation of the event to a large subregion of ${\cal A}$. On the
other hand, from Section~\ref{sec:average_time_to_false_isolation}, we
will see that a large $\underline{\omega}_0$ decreases the probability
of false isolation, a performance metric of change detection/isolation
procedure, which we define in Section~\ref{sec:problem_formulation}.

We define the {\em influence--region} of sensor $s$ as
$\mathcal{T}^{(s)} \ := \ \{x \in \mathcal{A} : \|\ell^{(s)}-x\|
\leqslant \InfRange    \}$. For the Boolean sensing model, $\InfRange
= \DetRange$, and hence, ${\cal D}^{(s)} = {\cal T}^{(s)}$ for all $1
\leqslant s \leqslant n$, and for the power law path--loss sensing
model, $\InfRange > \DetRange$, and hence, ${\cal D}^{(s)} \subset {\cal
T}^{(s)}$ for all $1 \leqslant s \leqslant n$ 
(see Fig.~\ref{fig:tfi_coverage}).

\begin{figure}[t]
\centering
\subfigure[Detection and influence regions of the Boolean model]
{ 
\includegraphics[width = 45mm, height = 45mm]{coverage_new}
\label{fig:coverage_new_left}
}
	\hspace{10mm}
\subfigure[Detection and influence regions of the power law path loss model]
{ 
\includegraphics[width = 45mm, height = 45mm]{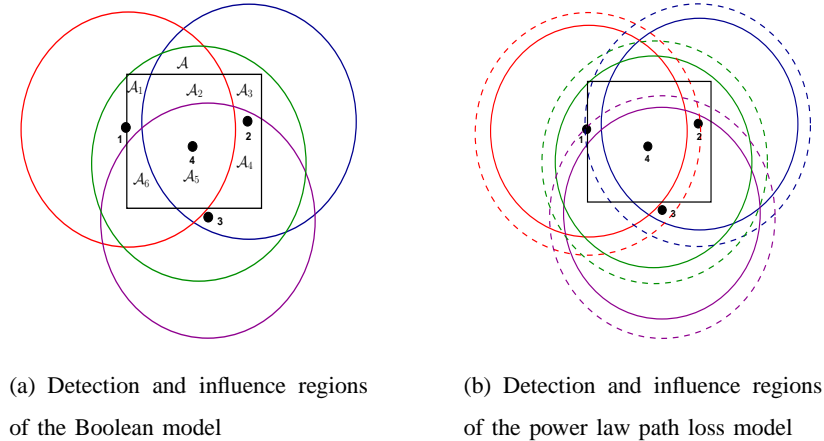}
}
\caption{{\bf Influence and detection regions}:
         A simple example of partitioning of $\mathcal{A}$ in a large
		 $\wsn$. The coloured solid circles around each sensor node denote
		 their detection regions. The four sensor nodes, in the figure, 
		 divide the $\ROI$, indicated by the square region, into regions
		 $\mathcal{A}_1, \cdots, \mathcal{A}_6$ such that region
		 $\mathcal{A}_i$ is detection--covered by a unique set of
		 sensors $\mathcal{N}_i$. The
		 dashed circles represent the influence regions. In the Boolean
		 model, the influence region of a sensor coincides with
		 its detection region.
         }
\label{fig:tfi_coverage}
\end{figure}
Recalling the sets of sensors ${\cal N}_i$, $1 \leqslant i \leqslant N$,
defined in Section~\ref{sec:detection-range}, we define the {\em
influence region of the set of sensors} $\mathcal{N}_i$ as the region
$\mathcal{B}_i$ such that each $x \in \mathcal{B}_i$ is within the
influence range of all the sensors in $\mathcal{N}_i$, i.e.,
$\mathcal{B}_i \ := \ {\cal B}({\cal N}_i) \ 
 :=  \bigcap_{s \in \mathcal{N}_i } \mathcal{T}^{(s)}$.  
Note that $\mathcal{A}(\mathcal{N}_i) =  \left(\underset{s \in
\mathcal{N}_i }{\bigcap} \mathcal{D}^{(s)}\right) \bigcap
\left(\underset{s' \notin \mathcal{N}_i }{\bigcap}
\overline{\mathcal{D}^{(s')}}\right)$, where $\overline{{\cal D}}$ is
the complement of the set ${\cal D}$, and ${\cal D}^{(s)} \subseteq
{\cal T}^{(s)}$. Hence, $\mathcal{A}(\mathcal{N}_i) \subseteq
\mathcal{B}(\mathcal{N}_i)$. For the power law path--loss sensing model,
${\cal D}^{(s)} \subset {\cal T}^{(s)}$ for all $1 \leqslant s \leqslant
n$, and hence, $\mathcal{A}(\mathcal{N}_i) \subset
\mathcal{B}(\mathcal{N}_i)$ for all $1 \leqslant i \leqslant N$. For the
Boolean sensing model,
$\mathcal{A}(\mathcal{N}_i) =  
\mathcal{B}(\mathcal{N}_i)   
\bigcap
\left(\underset{s' \notin \mathcal{N}_i }{\bigcap}
\overline{\mathcal{D}^{(s')}}\right)$, and hence 
$\mathcal{A}(\mathcal{N}_i) =
\mathcal{B}(\mathcal{N}_i)$ only when ${\cal N}_i = \{1,2,\cdots,n\}$.
Thus, for a general sensing model, $\mathcal{A}(\mathcal{N}_i) \subseteq
\mathcal{B}(\mathcal{N}_i)$. We note here that 
in the Boolean and the power law path loss models, 
an event which does not lie in the detection subregion of ${\cal N}_i$,
but lies in its influence subregion (i.e., $\ell_e \in
\mathcal{B}(\mathcal{N}_i)\setminus\mathcal{A}(\mathcal{N}_i)$) can
be detected due to ${\cal N}_i$ because of the stochastic nature
of the observations; in the power law path loss sensing model, this is
also because of  the difference in losses $\rho(d_{e,s})$ between different
sensors.

\noindent
{\bf Remark:} The definition of the 
detection and influence ranges have involved two design
parameters $\mu_1$ and $\underline{\omega}_0$ which can be used to
``tune'' the performance of the distributed detection schemes
that we develop.

\hfill\qed

\vspace{-4mm}

\subsection{Isolating the Event}
\label{subsec:discussion}
In Section II D, we provided an outline of a class of distributed
detection procedures that will yield a stopping rule. On stopping, a
decision for the location of the event is made, which is called {\em
isolation.} In 
Section~\ref{sec:distributed_change_detection_isolation_procedures}, we
will provide specific distributed detection/isolation procedures in
which stopping will be due to one of the sensor sets ${\cal N}_i$.

An event occurring at location $\ell_e \in {\cal A}_i$ can influence
sensors $s'$ which influence cover $\ell_e$, and hence, the detection
can be due to sensors ${\cal N}_i \neq {\cal N}_j$ which influence cover
$\ell_e$. Thus, we isolate the event to the influence region of the
sensors that detect the event. 
Because of noise,
detection can be due to a sensor set ${\cal N}_{h}$ which does not
influence cover the event. Such an error event is called false
isolation. 

An event occurring at $\ell_e \in {\cal
A}_i$ is influence covered by sensors $s' \in {\cal N}(\ell_e) : = \{s:
\|\ell^{(s)}-\ell_e\| \leq \InfRange\}$. Hence, the detection due to
any ${\cal N}_j \subseteq {\cal N}(\ell_e)$ corresponds to the isolation of the event, and that due to  
${\cal N}_j \not\subseteq {\cal N}(\ell_e)$ corresponds
to false isolation. Note that in the case of Boolean sensing model 
${\cal N}(\ell_e) = {\cal N}_i$.

In Section~\ref{sec:problem_formulation}, we formulate the problem of
quickest detection of an event and {\em isolating the event to one of
the influence subregions $\mathcal{B}_1,
\mathcal{B}_2,\cdots,\mathcal{B}_N$} under a false alarm and false
isolation constraint. 

\section{Problem Formulation}
\label{sec:problem_formulation}
We are interested in studying the {\em problem of distributed event
detection/isolation} in the setting developed in
Section~\ref{sec:system_model}. Given a sample node deployment (i.e.,
given ${\bm\ell}$), and {\em having chosen a value of the detection
range, $\DetRange$}, we partition the $\ROI$, $\mathcal{A}$ into the
detection--subregions, $\mathcal{A}_1, \mathcal{A}_2, \cdots,
\mathcal{A}_N$. Let $\mathcal{N}_i$ be the set of sensors that
detection--cover the region $\mathcal{A}_i$. Having chosen the influence
range $\InfRange$, the influence region $\mathcal{B}_i$ of the set of
sensor nodes $\mathcal{N}_i$ can be obtained. We define the following
set of hypotheses
\begin{eqnarray*}
{\bf H}_0     & : & \text{event not occurred}, \\ 
{\bf H}_{T,i} & : & \text{event occurred at time $T$ in subregion} 
                    \ \mathcal{A}_i, \ \ T=1,2,\cdots, \ i=1,2,\cdots,N.
\end{eqnarray*}

The event occurs in one of the detection subregions ${\cal A}_i$, but we
will only be able to isolate it to one of the influence subregions
${\cal B}_i$ that is consistent with the ${\cal A}_i$ (see 
Section~\ref{subsec:discussion}).  We study
distributed procedures described by a stopping time $\tau$, and an
isolation decision $L(\tau) \in \{1,2,\cdots,N\}$ (i.e., the tuple
$(\tau, L)$) that detect an event at time $\tau$ and locate it to
$L(\tau)$ (i.e., to the influence region $\mathcal{B}_{L(\tau)}$) subject to a false alarm
and false isolation constraint.  The {\em false alarm constraint}
considered is the average run length to false alarm  $\tfa$, and the
{\em false isolation constraint} considered is the probability of false
isolation  $\pfi$, each of which we define as follows.
\begin{definition} 
\label{def:overall_tfa}
The {\bf Average Run Length to False Alarm} $\tfa$ of a change
detection/isolation procedure $\tau$ is defined as the
expected number of samples taken under null hypothesis ${\bf H}_0$
to raise an alarm, i.e., 
\begin{eqnarray*}
\tfa(\tau) & := & {\sf E}_\infty\left[\tau\right], 
\end{eqnarray*}
where ${\sf E}_\infty[\cdot]$ is the expectation operator (with the
corresponding probability measure being ${\sf P}_\infty\{\}$) when the
change occurs at infinity. 
\qed
\end{definition}

\begin{definition} 
\label{def:overall_pfi}
The {\bf Probability of False Isolation} $\pfi$ of a change
detection/isolation procedure $\tau$ is defined as the
supremum of the probabilities of making an incorrect isolation decision,
i.e.,
\begin{eqnarray*}
\pfi(\tau) \ := \
\max_{1 \leq i \leq N} \ 
 \sup_{\ell_e \in {\cal A}_i} \ 
\max_{1 \leq j \leq N, {\cal N}_j \not\subseteq{\cal N}(\ell_e)} \
{\mathsf P}_1^{({\bf d}(\ell_{e}))}\left\{L(\tau) = j \right\}
\end{eqnarray*}
where we recall that ${\cal N}(\ell_e) = \{s:\|\ell^{(s)}-\ell_e\|\leq
\InfRange\}$ is the set of sensors that influence covers $\ell_e
\in {\cal A}_i$.
\qed
\end{definition}
In the case of Boolean sensing model, the post--change pdfs depend only
on the index $i$ of the detection subregion where the event occurs, and
hence, the $\pfi$ is given by 
\begin{eqnarray*}
\pfi(\tau) \ := \
\max_{1 \leq i \leq N} \ 
\max_{1 \leq j \leq N, {\cal N}_j \not\subseteq{\cal N}_i} \
{\mathsf P}_1^{(i)}\left\{L(\tau) = j \right\}.
\end{eqnarray*}
In  
\cite{nikiforov03lower-bound-for-det-isolation}, 
Nikiforov defined the probability of false isolation, also, over the set
of all possible change times, as 
${\sf SPFI}(\tau) 
\ := \ 
\sup_{1 \leq i \leq N} \ \ 
\sup_{1 \leq j\neq i \leq N} \ \
\sup_{T \geq 1} \ {\mathsf P}_T^{(i)}\left\{L(\tau) = j \mid \tau \geq
T\right\}$.
Define the following classes of change detection/isolation procedures,
\begin{eqnarray*}
{\Delta}(\gamma,\alpha) &:=& \left\{(\tau,L): \tfa(\tau) \geq \gamma,
{\sf SPFI}(\tau) \leqslant \alpha \right\}, \\
\widetilde\Delta(\gamma,\alpha) &:=& \left\{(\tau,L): \tfa(\tau) \geq \gamma,
\pfi(\tau) \leqslant \alpha \right\}.
\end{eqnarray*} 

We define the supremum average
detection delay $\add$ performance for the procedure $\tau$, in the same
sense as Pollak \cite{pollak85} (also see
\cite{nikiforov03lower-bound-for-det-isolation}), as the maximum mean 
number of samples taken under any hypothesis ${\bf H}_{T,i}, \
i=1,2,\cdots,N$, to raise an alarm, i.e.,
\begin{eqnarray*}
\add(\tau) := 
\underset{\ell_e \in \mathcal{A}}{\sup} \ 
\underset{T \geqslant 1}{\sup} \ {\mathsf E}^{({\bf
d}({\ell_e}))}_T\left[\tau-T|\tau \geq T\right].
\end{eqnarray*}
We are interested in obtaining an
optimal procedure $\tau$ that minimises the $\add$ subject to the
average run length to false alarm and the probability of false isolation
constraints,
\begin{eqnarray*}
\inf             
& & 
    \underset{\ell_e \in \mathcal{A}}{\sup} \ 
    \underset{T \geqslant 1}{\sup} \ 
    {\mathsf E}^{({\bf d}(\ell_e))}_T\left[\tau-T|\tau \geq T \right]\\
\text{subject to} & &
\begin{array}{rcl}
\tfa(\tau) & \geqslant & \gamma\\ 
\pfi(\tau) & \leqslant & \alpha. 
\end{array}
\end{eqnarray*}

The change detection/isolation problem that we pose here is motivated
by the framework of 
\cite{nikiforov95change_isolation},
\cite{nikiforov03lower-bound-for-det-isolation},
\cite{tartakovsky08multi-decision},
which we discuss in the next subsection.

\subsection{Centralised Recursive Solution for the Boolean Sensing Model}
In \cite{nikiforov03lower-bound-for-det-isolation}, Nikiforov and in
\cite{tartakovsky08multi-decision}, Tartakovsky studied a change
detection/isolation problem that involves $N > 1$ post--change
hypotheses (and one pre--change hypothesis). Thus, their formulation can
be applied to our problem. But, in their model, 
the pdf of ${\bf X}_k$ for $k \geq T$, 
under hypothesis ${\bf H}_{T,i}$, 
 $g_i$ is completely known. It should be
noted that in our problem, in the case of power law path--loss sensing
model, the pdf of the observations under any post--change hypothesis is
unknown as the location of the event is unknown. The problem posed by
Nikiforov 
\cite{nikiforov03lower-bound-for-det-isolation} is
\begin{eqnarray}
\label{eqn:problem1}
\inf_{(\tau,L) \in {\Delta}(\gamma,\alpha)} 
& & 
    \underset{1 \leqslant i \leqslant N}{\sup} \ 
    \sup{\mathsf E}^{(i)}_T\left[\tau-T|\tau \geq T\right],
\end{eqnarray}
and that by Tartakovsky  \cite{tartakovsky08multi-decision} is
\begin{eqnarray}
\label{eqn:problem2}
\inf_{(\tau,L)\in\widetilde{\Delta}(\gamma,\alpha)} 
& & 
    \underset{1 \leqslant i \leqslant N}{\sup} \ 
    \sup{\mathsf E}^{(i)}_T\left[\tau-T|\tau \geq T\right].
\end{eqnarray}
Nikiforov 
\cite{nikiforov03lower-bound-for-det-isolation} and Tartakovsky  
\cite{tartakovsky08multi-decision} obtained 
asymptotically optimal {\em centralised change detection/isolation}
procedures as $\min\{\gamma,\frac{1}{\alpha}\} \to \infty$, the $\add$
of which is given by the following theorem.

\begin{theorem}[Nikiforov 03]
\label{thm:niki}
For the $N$--hypotheses change detection/isolation problem (for the Boolean
sensing model) defined in 
Eqn.~\eqref{eqn:problem1}, the asymptotically maximum mean delay optimal 
detection/isolation procedure 
$\tau^{\sf *}$ has the property, 
\begin{eqnarray*}
\label{eqn:niki-add}
\add(\tau^{\sf *}) & \stackrel{\leq}{\sim} & \max\left\{\frac{\ln \gamma}{
\underset{1\leqslant i \leqslant N}{\min} \ \ \ \mbox{KL}(g_i,g_0)
	}, \frac{-\ln(\alpha)}{
\underset{
	1 \leqslant i \leqslant N, 
	1 \leqslant j \neq i
\leqslant N}{\min} \ \ \mbox{KL}(g_i,g_j)
		}\right\}, \ \ \text{as} \ \min\left\{\gamma,\frac{1}{\alpha}\right\} \to \infty, 
\end{eqnarray*}
where $\mbox{KL}(\cdot,\cdot)$ is the Kullback--Leibler divergence
function, and  $g_i$ is the pdf of the
observation ${\bf X}_k$ for $k \geq T$ under hypothesis ${\bf H}_{T,i}$. 
\qed
\end{theorem}

\noindent
{\bf Remark:}
Since, 
${\Delta}(\gamma,\alpha) \subseteq 
\widetilde{\Delta}(\gamma,\alpha)$, the
asymptotic upper bound on $\add$ for $\tau^*$ is also an upper bound
for the $\add$ over the set of procedures in $\widetilde{\Delta}(\gamma,\alpha)$. 

In the case of Boolean sensing model, for any post--change hypothesis
${\bf H}_{T,i}$,  only the set of sensor nodes that
detection cover (which is the same as influence cover) the subregion
${\cal A}_i$ switch to a post--change pdf $f_1$ (and the distribution of
other sensor nodes continues to be $f_0$). Since the pdf of the sensor
observations are conditionally i.i.d., the pdf of the observation
vector, in the Boolean sensing model,  corresponds to the post--change
pdf $g_i$ of the centralised problem studied by Nikiforov 
\cite{nikiforov03lower-bound-for-det-isolation} and by Tartakovsky  
\cite{tartakovsky08multi-decision}.  
Thus, their problem directly applies to our setting with the Boolean
sensing model. In our work, however, we propose algorithms for the
change detection/isolation problem for the power law sensing model as
well. Also, the procedures proposed by Nikiforov and by Tartakovsky are
(while being recursive) centralised, whereas we propose distributed procedures which are
computationally simple. 

In Section~\ref{sec:distributed_change_detection_isolation_procedures},
we propose distributed detection/isolation procedures $\sf{MAX}$,
$\sf{HALL}$ and $\sf{ALL}$ and analyse their false alarm ($\tfa$), false
isolation ($\pfi$) and the detection delay ($\add$) properties.

\section{Distributed Change Detection/Isolation Procedures}
\label{sec:distributed_change_detection_isolation_procedures}
In this section, we study the procedures $\sf{MAX}$ and $\sf{ALL}$ for
change detection/isolation in a distributed setting. Also, we propose a 
distributed detection procedure ``${\sf HALL}$,'' and 
analyse the $\add$, the $\tfa$, and the $\pfi$ performance. 

\subsection{The $\sf{MAX}$ Procedure}
Tartakovsky and Veeravalli proposed a decentralised procedure $\sf{MAX}$
for a collocated scenario in
\cite{stat-sig-proc.tartakovsky-veeravalli03quickest-change-detection}.
We extend the $\sf{MAX}$ procedure to a large $\wsn$ under the $\tfa$
and $\pfi$ constraints. Recalling Section~\ref{sec:system_model}, each
sensor node $i$ employs $\CUSUM$ for local change detection between pdfs
$f_0$ and $f_1(\cdot;\DetRange)$. Let $\tau^{(i)}$ be the random time at
which the $\CUSUM$ statistic of sensor node $i$ crosses the threshold
$c$ for the first time. At each time $k$, the local decision of sensor
node $i$, $D_k^{(i)}$ is defined as 
\begin{eqnarray*} 
D_k^{(i)} 
& := & \left\{
       \begin{array}{ll}
	   0, &  \text{for} \ k < \tau^{(i)}\\
	   1, &  \text{for} \ k \geqslant  \tau^{(i)}.
	   \end{array}
	   \right.
\end{eqnarray*} 
The global decision rule $\tau^{\sf{MAX}}$ declares an alarm at the
earliest time slot $k$ at which all sensor nodes $j \in \mathcal{N}_i$
for some $i=1,2,\cdots,N$ have crossed the threshold $c$. Thus,
\begin{eqnarray*} 
\tau^{\mathsf{MAX},(\mathcal{N}_i)} 
& := & \inf\left\{k: D_k^{(j)} = 1, \ \forall j \in \mathcal{N}_i\right\}
\  = \ \min \left\{\tau^{(j)} : j \in \mathcal{N}_i \right\}\\
\tau^{\mathsf{MAX}} & := & \min\left\{\tau^{\mathsf{MAX},(\mathcal{N}_i)}: 1 \leqslant i \leqslant N \right\}.  
      \end{eqnarray*} 
i.e., the {\sf MAX} procedure declares an alarm at the earliest time 
instant when the $\CUSUM$ statistic of all the sensor nodes ${\cal
N}_i$ corresponding to hypothesis ${\bf H}_{T,i}$ of some $i$ have crossed the
threshold at least once. 
The isolation rule is 
$L(\tau)  =  \arg\min_{1\leq i\leq N} \{\tau^{\sf MAX, ({\cal N}_i)}\}$, 
i.e., to declare that the event has occurred in the
influence region ${\cal B}_{L(\tau)} = \mathcal{B}(\mathcal{N}_{L(\tau)})$ 
corresponding to the set of sensors $\mathcal{N}_{L(\tau)}$ that raised the alarm.

\subsection{$\sf{ALL}$ Procedure}
Mei, \cite{stat-sig-proc.mei05information-bounds}, and 
Tartakovsky and Kim, \cite{stat-sig-proc.tartakovsky-kim06decentralized}, 
proposed a decentralised procedure $\sf{ALL}$, again for a collocated 
network. We extend the $\sf{ALL}$ procedure to a large extent network 
under the $\tfa$ and the $\pfi$ constraints. Here, each sensor node $i$ 
employs $\CUSUM$ for local change detection between pdfs $f_0$ and 
$f_1(\cdot;\DetRange)$. Let $C_k^{(i)}$ be the $\CUSUM$ statistic of sensor 
node $i$ at time $k$. {\em The $\CUSUM$ in the sensor nodes is allowed 
to run freely even after crossing the threshold $c$}. Here, the local 
decision of sensor node $i$ is 
\begin{eqnarray*} 
D_k^{(i)} & := & \left\{
\begin{array}{ll}
0, &  \text{if} \ C_k^{(i)} < c\\
1, &  \text{if} \ C_k^{(i)} \geqslant  c.
\end{array}
\right.
\end{eqnarray*} 
The global decision rule $\tau^{\sf{ALL}}$ declares an alarm at 
the earliest time slot $k$ at which the local decision of all the 
sensor nodes corresponding to a set $\mathcal{N}_i$, for some 
$i=1,2,\cdots,N$, are 1, i.e.,
\begin{eqnarray*} 
\tau^{\mathsf{ALL},(\mathcal{N}_i)} 
& := & \inf\left\{k: D_k^{(j)} = 1, \ \forall j \in \mathcal{N}_i\right\} 
 \ = \ \inf \left\{ k: C_k^{(j)} \geqslant c, \forall j \in \mathcal{N}_i \right\}\\
\tau^{\mathsf{ALL}} & := & \min\left\{\tau^{\mathsf{ALL},(\mathcal{N}_i)}: 1 \leqslant i \leqslant N \right\}.  
      \end{eqnarray*} 
The isolation rule is 
$L(\tau) \ = \ \arg\min_{1\leq i\leq N} \{\tau^{\sf ALL, ({\cal
N}_i)}\}$, 
i.e., to declare that the event has occurred in the
influence region ${\cal B}_{L(\tau)} = \mathcal{B}(\mathcal{N}_{L(\tau)})$ 
corresponding to the set of sensors $\mathcal{N}_{L(\tau)}$ that raised the alarm.

\subsection{$\sf{HALL}$ Procedure}
\begin{figure}[t]
\begin{center}
\includegraphics[width=3.5in]{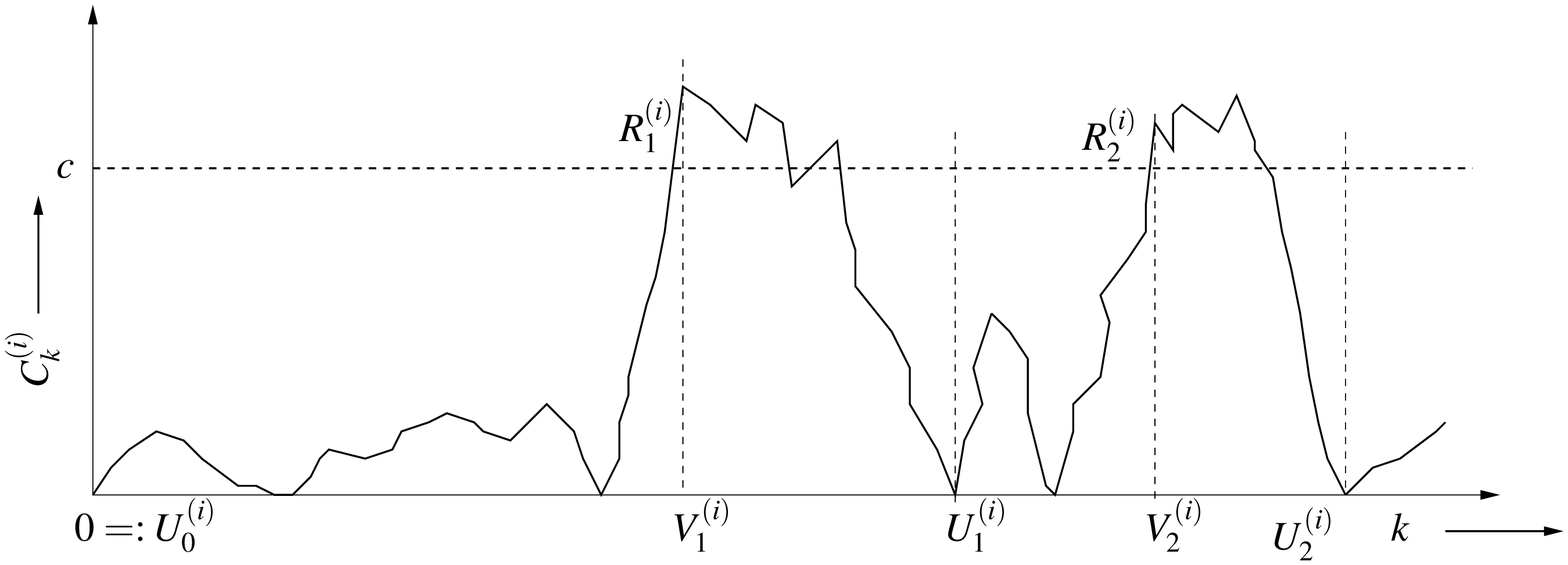}
\caption{{\sf ALL} and {\sf HALL}: Evolution of $\CUSUM$ statistic 
$C_k^{(i)}$ of node $i$ plotted vs.\ $k$. 
Note that at time $k = V_j^{(i)}$, $R_j^{(i)}$ is the excess above the threshold. 
} 
\label{fig:single_cusum}
\end{center}
\end{figure}
Motivated by ${\sf ALL}$, and the fact that sensor noise can make the $\CUSUM$ statistic fluctuate
around the threshold, we propose a local decision rule which is 0 when
the $\CUSUM$
statistic has visited zero and has not crossed the threshold yet and is 1
otherwise. We explain the $\sf{HALL}$ procedure below. 

\noindent
The following discussion is illustrated in Fig.~\ref{fig:single_cusum}.
Each sensor node $i$ computes a $\CUSUM$ statistic $C_k^{(i)}$ based on
the LLR of its own observations between the pdfs $f_1(\cdot;\DetRange)$ and
$f_0$. Define $U_0^{(i)} := 0$. Define $V_1^{(i)}$ as the time at which
$C_k^{(i)}$ crosses the threshold $c$ (for the first time) as: 
\begin{eqnarray*}
V_1^{(i)}  :=  \inf \left\{k: C_k^{(i)} \geqslant c\right\}
\end{eqnarray*}
(see Fig.~\ref{fig:single_cusum} where the ``overshoots'' $R_k^{(i)}$, at
$V_k^{(i)}$, are also shown). Note that $\inf\emptyset := \infty$. 
Next define
\begin{eqnarray*}
U_1^{(i)}  :=  \inf \left\{k > V_1^{(i)}: C_k^{(i)} = 0\right\}. 
\end{eqnarray*}
Now starting with $U_1^{(i)}$, we can recursively define $V_2^{(i)},U_2^{(i)}$ 
etc. in the obvious manner (see Fig.~\ref{fig:single_cusum}). 
Each node $i$ computes the local decision $D_k^{(i)}$ based on the
$\CUSUM$ 
statistic $C_k^{(i)}$ as follows: 
\begin{eqnarray}   
\label{eqn:ld}
D_k^{(i)} & = & \left\{
\begin{array}{lll}
1, &  \text{if} \ V^{(i)}_j \leqslant k < U^{(i)}_j \ \text{ for some } j\\ 
0, &  \text{otherwise.} 
\end{array}
\right.
\end{eqnarray}   
\vspace{0mm}
The global decision rule
 is a stopping time $\tau^{\sf{HALL}}$ defined 
as the earliest time slot $k$ at which all the sensor nodes in a region have a  
local decision $1$, i.e.,
\begin{eqnarray*} 
\tau^{\mathsf{HALL},(\mathcal{N}_i)} & := & \inf\left\{k: D_k^{(j)} = 1, \ \forall j \in \mathcal{N}_i\right\},\\
\tau^{\mathsf{HALL}} & := & \min\left\{
	\tau^{\mathsf{HALL},(\mathcal{N}_i)} :  1 \leqslant i \leqslant N\right\}. 
      \end{eqnarray*} 
The isolation rule is 
$L(\tau)  =  \arg\min_{1\leq i\leq N} \{\tau^{\sf HALL, ({\cal N}_i)}\}$, 
i.e., to declare that the event has occurred in the
influence region ${\cal B}_{L(\tau)} = \mathcal{B}(\mathcal{N}_{L(\tau)})$ 
corresponding to the set of sensors $\mathcal{N}_{L(\tau)}$ that raised the alarm.

\noindent
{\bf Remark:}
The procedures $\sf{HALL}, \sf{MAX}$ and $\sf{ALL}$ differ only in their
local decision rule; the global decision rule as a function of
$\{D_k^{(i)}\}$ is the same for  $\sf{HALL}, \sf{MAX}$ and $\sf{ALL}$.
For the distributed procedures $\sf{MAX}$, $\sf{ALL}$, and $\sf{HALL}$,
we analyse the $\tfa$ in Section \ref{sec:false_alarm_analysis}, the
$\pfi$ in Section \ref{sec:average_time_to_false_isolation}, and the
$\add$ performance in Section \ref{sec:average_detection_delay}. 

\subsection{Average Run Length to False Alarm ($\tfa$)}
\label{sec:false_alarm_analysis}

From the previous sections, we see that the stopping time of any
${\sf procedure}$ ($\mymax$, $\myall$, or ${\sf HALL}$) is the minimum of the stopping times corresponding to each
${\cal N}_i$, i.e., 
\begin{eqnarray*} 
\tau^{\mathsf{procedure}} & := &  
\min\left\{\tau^{\mathsf{procedure}, (\mathcal{N}_i)} :  1 \leqslant i \leqslant N\right\}. 
\end{eqnarray*} 
Under the null hypothesis ${\bf H}_0$, the $\CUSUM$ statistics $C_k^{(s)}$s
of sensors $s \in {\cal N}_i$ are driven by independent noise processes,
and hence, $C_k^{(s)}$s are independent. But, there can be a sensor
that is common to two different ${\cal N}_i$s, and hence,  
$\tau^{\mathsf{procedure}, (\mathcal{N}_i)}$s, in general, are not independent.  
We provide asymptotic lower bounds for the $\tfa$ for  
$\sf{MAX}$,
$\sf{HALL}$, and $\sf{ALL}$, in the following theorem. 
\begin{theorem}
For local $\CUSUM$ threshold $c$, 
\label{thm:tfa}
\begin{eqnarray}
\tfa(\tau^{\sf MAX})  & \geqslant & \exp\left(a_{\sf MAX} c\right) \cdot (1+o(1)) \\ 
\tfa(\tau^{\sf HALL}) & \geqslant & \exp\left(a_{\sf HALL} c\right) \cdot (1+o(1)) \\ 
\tfa(\tau^{\sf ALL})  & \geqslant & \exp\left(a_{\sf ALL} c\right) \cdot (1+o(1)) 
\end{eqnarray}
($o(1) \to 0$ as $c \to \infty$), where 
for any arbitrarily small $\delta > 0$, 
$a_{\sf MAX} = a_{\sf HALL} = 1- \delta$, $a_{\sf ALL} = m -
\delta$, 
\end{theorem}
where $m = \min\{{\cal N}_i \setminus \bigcup_{\substack{j \neq
i, \ }{j \in {\cal I}}} {\cal N}_j : i \in {\cal I}\}$, 
${\cal I}$ is the set of indices of the detection sets that are 
minimal in the partially order of set inclusion among the detection
sets.
\begin{proof}
See Appendix~\ref{app:tfa}.
\end{proof}
Thus, for ${\sf procedure}$, for a given $\tfa$ requirement of $\gamma$,
it is sufficient to choose the threshold $c$ as
\begin{eqnarray}
\label{eqn:c_for_gamma}
 c& = \frac{\ln\gamma}{a_{\sf procedure}} (1+o(1)), \ \ \text{as $\gamma
 \to \infty$}.
\end{eqnarray}


\vspace{-3mm}

\subsection{Probability of False Isolation ($\pfi$)}
\label{sec:average_time_to_false_isolation}
A false isolation occurs when the hypothesis ${\bf H}_{T,i}$ is true for
some $i$ and the hypothesis ${\bf H}_{T,j} \neq {\bf H}_{T,i}$ is declared to be
true at the time of alarm, {\em and} the event does not lie in the region
$\mathcal{B}(\mathcal{N}_j)$. The following theorem provide asymptotic
upper bounds for the $\pfi$ for each of the procedures $\mymax$,
$\myall$, and $\myhall$.

\begin{theorem}
\label{thm:pfi}
For local $\CUSUM$ threshold $c$,

\begin{eqnarray}
{\sf PFI}(\tau^{\mymax}) &\leq& 
\frac{\exp\left(-b_{\sf MAX} c\right)}{B_{\sf MAX}} \cdot (1+o(1))\\
{\sf PFI}(\tau^{\myhall}) &\leq& 
\frac{\exp\left(-b_{\sf HALL} c\right)}{B_{\sf HALL}} \cdot (1+o(1))\\
{\sf PFI}(\tau^{\myall}) &\leq& 
\frac{\exp\left(-b_{\sf ALL} c\right)}{B_{\sf ALL}} \cdot (1+o(1)).
\end{eqnarray}
where $o(1) \to 0$ as $c \to \infty$, and $
b_{\sf MAX} = b_{\sf HALL} =
\frac{m\underline{\xi \omega}_0}{2}-\frac{1+\bar{m}}{n}$, 
$b_{\sf ALL}= \frac{m\underline{\xi \omega}_0}{2}-\frac{1}{n}$, 
$\underline{\omega}_0 = 1$ for Boolean sensing model,
$\xi$ is 2 for Boolean sensing model and is 1 for path--loss sensing
model,
$m = \min\left\{
|{\cal N}_j \setminus {\cal N}(\ell_e)| :
{1 \leq i \leq N, \ell_e \in {\cal A}_i, 1 \leq j \leq N, {\cal N}_j
\not\subseteq {\cal N}(\ell_e)} \right\}$
and 

\noindent
$\bar{m} = \max\left\{
|{\cal N}_j \setminus {\cal N}(\ell_e)| :
{1 \leq i \leq N, \ell_e \in {\cal A}_i, 1 \leq j \leq N, {\cal N}_j
\not\subseteq {\cal N}(\ell_e)} \right\}$,
and $B_{\sf MAX}$,
$B_{\sf HALL}$, and $B_{\sf ALL}$ are  positive constants. 
\end{theorem}
\begin{proof}
See Appendix~\ref{app:pfi}.
\end{proof}
Thus, for a given $\pfi$ requirement of $\alpha$, the threshold $c$ for
should satisfy 
\begin{eqnarray}
\label{eqn:c_for_alpha}
c \ = \ \frac{-\ln B_{\sf procedure} -\ln\alpha}{b_{\sf procedure}} (1+o(1)) 
   \ = \ \frac{-\ln\alpha}{b_{\sf procedure}} (1+o(1)), \ \ \  \text{as $\alpha \to 0$}.
\end{eqnarray}

\subsection{Supremum Average Detection Delay (\textsf{SADD})}
\label{sec:average_detection_delay}
In this section, we analyse the $\add$ performance of the distributed
detection/isolation procedures. We observe that for any sample path of
the observation process, for the same threshold $c$, the ${\mymax}$
rule raises an alarm first, followed by the ${\myhall}$ rule, and then
by the ${\myall}$ rule. This ordering is due to the following reason.
For each sensor node $s$, let $\tau^{(s)}$ 
be the {\em first time instant} at
which the $\CUSUM$ statistic $C_k^{(s)}$ crosses the threshold $c$
(denoted by $V_1^{(i)}$ in Figure~\ref{fig:single_cusum}).
Before time $\tau^{(s)}$, the local decision is 0 for all 
the procedures, ${\sf MAX}$, ${\sf ALL}$, and ${\sf HALL}$. For
${\mymax}$, for all $k \geq \tau^{(s)}$, the local decision
$D_{k}^{(s)} = 1$. Thus, the stopping time of  $\mymax$ is at least as
early as that of ${\myhall}$ and $\myall$. The local decision of
$\myall$ is 1 ($D_k^{(s)} = 1$) only at those times $k$ for which $C_k^{(s)}
\geq c$. However, even when $C_k^{(s)} < c$, the local decision of
$\myhall$ is 1 if $V_j^{(s)} \leq k < U_j^{(s)}$ 
(see Figure~\ref{fig:single_cusum}) for some $j$. Thus, the local decisions of ${\sf
MAX}$, ${\sf HALL}$, and ${\sf ALL}$ are ordered as, for all $k \geq 1$,
$ D_k^{(s)}(\mathsf{MAX})   \geqslant 
  D_k^{(s)}(\mathsf{HALL})  \geqslant  
  D_k^{(s)}(\mathsf{ALL})$, 
and hence, 
$\tau^{\mathsf{MAX},({\cal N}_i)}  \leqslant 
\tau^{\mathsf{HALL},{(\cal N}_i)}  \leqslant 
\tau^{\mathsf{ALL},{(\cal N}_i)}$. 
Each of the stopping times
$\sf{MAX}$, $\sf{HALL}$, or $\sf{ALL}$ is the minimum of stopping times
corresponding to the sets of sensors $\{\mathcal{N}_i : i =
1,2,\cdots,N\}$, i.e., 
\begin{eqnarray*}
\tau^{\mathsf{procedure}} & =  &
\min\{\tau^{\mathsf{procedure},(\mathcal{N}_i)}:i=1,2,\cdots,N\}
\end{eqnarray*}
where ``$\mathsf{procedure}$'' can be $\mathsf{MAX}$ or $\mathsf{HALL}$ or $\mathsf{ALL}$. 
Hence, we have
\begin{eqnarray}
\label{o6_eqn:ordered_stopping} 
\tau^{\mathsf{MAX}} \  \leqslant \ \tau^{\mathsf{HALL}} \ \leqslant \ \tau^{\mathsf{ALL}}.
\end{eqnarray}
From \cite{stat-sig-proc.mei05information-bounds}, we see that  
\begin{eqnarray}
\sup_{T \geqslant 1} \
{\mathsf E}_T^{(i)}\left[\tau^{\mathsf{ALL},({\cal N}_i)}-T\mid
\tau^{\sf ALL, ({\cal N}_i)} \geq T \right]
&=& \frac{c}{I} \left(1+o(1)\right)
\end{eqnarray}
where $I$ is the
Kullback--Leibler divergence between the post--change and the
pre--change pdfs. 
For $\ell_e \in \mathcal{A}_i$, we have
$\forall s \in \mathcal{N}_i, \ d_{e,s} \leqslant \DetRange$. Also, since 
$\tau^{\mathsf{ALL}} \leqslant \tau^{\mathsf{ALL},({\cal N}_i)}$, we have
\begin{eqnarray}
\label{eqn:sadd_b}
\sup_{\ell_e \in {\cal A}_i} \ 
\sup_{T \geqslant 1} \
{\mathsf E}_T^{({\bf d}({\ell_e}))}\left[\tau^{\mathsf{ALL}}-T\mid \tau^{\sf ALL} \geq T \right] 
& \leqslant  & 
\sup_{\ell_e \in {\cal A}_i} \ 
\sup_{T \geqslant 1} \
{\mathsf E}_T^{({\bf d}({\ell_e}))}\left[\tau^{\mathsf{ALL},({\cal N}_i)}-T
\mid \tau^{\sf ALL} \geq T \right] 
\end{eqnarray}
From Appendix~\ref{app:sadd}, Eqn.~\eqref{eqn:sadd_b} becomes,
\begin{eqnarray}
\sup_{\ell_e \in {\cal A}_i} \ 
\sup_{T \geqslant 1} \
{\mathsf E}_T^{({\bf d}({\ell_e}))}\left[\tau^{\mathsf{ALL}}-T\mid \tau^{\sf ALL} \geq T \right] 
& \leqslant  & 
\sup_{\ell_e \in {\cal A}_i} \ 
\sup_{T \geqslant 1} \
{\mathsf E}_T^{({\bf d}({\ell_e}))}\left[\tau^{\mathsf{ALL},({\cal N}_i)}-T
\mid \tau^{\sf ALL, ({\cal N}_i)} \geq T \right]\nn
&=&  \frac{c}{\mathsf{KL}(f_1(\cdot;\DetRange),f_0)}(1+o(1)) 
\end{eqnarray}

From the above equation, and from Eqn.~\eqref{o6_eqn:ordered_stopping}, we have 
\begin{eqnarray}
\label{eqn:sadd_max_all_hall}
\add(\tau^{\mathsf{MAX}}) \  \leqslant \ \add(\tau^{\mathsf{HALL}}) \ \leqslant \ \add(\tau^{\mathsf{ALL}}) \
\leqslant \ \frac{c}{\mathsf{KL}(f_1(\cdot;\DetRange),f_0)}(1+o(1)), \ \text{as} \
c\to\infty, 
\end{eqnarray}

\noindent
{\bf Remark:}
Recall from Section~\ref{sec:detection-range} that $\mu_1 = 
h_e \rho(\DetRange)$. We now see that $\mu_1$ governs the detection 
delay performance, and $\mu_1$ can be chosen such that a requirement on $\add$ 
is met. Thus, to achieve a requirement on $\add$, we need to choose 
$\DetRange$ appropriately. A small value of $\DetRange$ (gives a large
$\mu_1$ and hence,) gives less detection delay 
compared to a large value of $\DetRange$. But, a small $\DetRange$ requires more 
sensors to detection--cover the $\ROI$.

In the next subsection, we discuss the asymptotic minimax delay 
optimality of the distributed procedures in relation to
Theorem~\ref{thm:niki}.

\subsection{Asymptotic Upper Bound on $\add$}
\label{sec:asymp_order_optimality}

For any change detection/isolation procedure to achieve a 
$\tfa$ requirement of $\gamma$ and $\pfi$ requirement of
$\alpha$, a threshold $c$ is chosen such that it satisfies 
Eqns.~\ref{eqn:c_for_gamma} and \ref{eqn:c_for_alpha}, i.e.,
\begin{eqnarray}
c &=& \max\left\{\frac{\ln\gamma}{a_{\sf procedure}}, 
                    \frac{-\ln\alpha}{b_{\sf procedure}} 
			\right\} (1+o(1)).
\end{eqnarray}
Therefore, from Eqn.\eqref{eqn:sadd_max_all_hall}, the $\add$ is given
by
\begin{eqnarray}
\label{eqn:order-optimal}
\add(\tau^{\mathsf{procedure}}) & \leqslant & 
 \frac{1}{\mathsf{KL}(f_1(\cdot;\DetRange),f_0)}
\cdot \max\left\{\frac{\ln\gamma}{a_{\sf procedure}}, 
                    \frac{-\ln\alpha}{b_{\sf procedure}} 
			\right\} (1+o(1)).
\end{eqnarray}
where $o(1) \to 0$ as $\min\{\gamma, \frac{1}{\alpha}\} \to \infty$.
Note that as $\DetRange$ decreases,
$\mathsf{KL}(f_1(\cdot;\DetRange),f_0) =
\frac{h_e^2\rho(\DetRange)^2}{2\sigma^2}$ increases.  
Thus, to achieve a smaller detection delay, the detection range
$\DetRange$ can be decreased, and the number of sensors $n$ can be
increased to cover the $\ROI$. 

We can compare the asymptotic $\add$ performance of the distributed
procedures $\sf{HALL}$, $\sf{MAX}$ and $\sf{ALL}$ against
Theorem~\ref{thm:niki} for the Boolean sensing model. 
For Gaussian pdfs $f_0$ and $f_1$, the KL divergence between the
hypotheses ${\bf H}_{T,i}$ and ${\bf H}_{T,j}$ is given by
\begin{eqnarray*}
\mathsf{KL}(g_i,g_j)
& = & \int \ln\left(\frac{
	\prod_{s \in {\cal N}_i} f_1(x^{(s)}) 
	\prod_{s' \notin {\cal N}_i} f_0(x^{(s')}) }{
{	\prod_{s \in {\cal N}_j} f_1(x^{(s)}) 
	\prod_{s' \notin {\cal N}_j} f_0(x^{(s')}) }
		}\right)  \  	\prod_{s \in {\cal N}_i} f_1(x^{(s)}) 
	\prod_{s' \notin {\cal N}_i} f_0(x^{(s')})  \ d{\bf x}\nn
& = & \int \left(\ln\left(
	\prod_{s \in {\cal N}_i}
	\frac{f_1(x^{(s)})}{f_0(x^{(s)})}\right)
	 - 
 \ln\left(
	{	\prod_{s \in {\cal N}_j } \frac{f_1(x^{(s)})}{f_0(x^{(s)})} }
	\right) \right) \  	\prod_{s \in {\cal N}_i} f_1(x^{(s)}) \prod_{s' \notin {\cal N}_i} f_0(x^{(s')})  \ d{\bf x} \nn 
 & = & \sum_{s \in {\cal N}_i} \mathsf{KL}(f_1,f_0)
 - \sum_{s \in {{\cal N}_j \cap \cal N}_i} \mathsf{KL}(f_1,f_0)
  + \sum_{s \in {\cal N}_j\setminus {\cal N}_i} \mathsf{KL}(f_1,f_0)\nn
  &=& |{\cal N}_i~\Delta~{\cal N}_j| \ \mathsf{KL}(f_1,f_0)
  \end{eqnarray*}
where the operator $\Delta$ represents the symmetric difference between
the sets. Thus, from Theorem~\ref{thm:niki} for Gaussian $f_0$ and
$f_1$, we have 
\begin{eqnarray*}
\add(\tau^{*}) &\leq& 
 \frac{1}{\mathsf{KL}(f_1,f_0)}
\cdot \max\left\{\frac{\ln\gamma}{a^*}, 
                    \frac{-\ln\alpha}{b^*} 
			\right\} (1+o(1)), \nn
\text{where} \ a^* &=& \min_{1 \leqslant i \leqslant N} |{\cal N}_i|, \nn 
\text{and}  \ b^* &=& \min_{\substack{1 \leqslant i \leqslant N\\
1 \leqslant j \leqslant N, \ {\cal N}_j \not\subseteq {\cal N}_i}} |{\cal
N}_i\Delta{\cal N}_j|. 
\end{eqnarray*}
The $\add$ performance of the distributed ${\sf
procedure}$ with the Boolean sensing model is 
\begin{eqnarray}
\label{eqn:order-optimal-1}
\add(\tau^{\mathsf{procedure}}) & \leqslant & 
 \frac{1}{\mathsf{KL}(f_1,f_0)}
\cdot \max\left\{\frac{\ln\gamma}{a_{\sf procedure}}, 
                    \frac{-\ln\alpha}{b_{\sf procedure}} 
			\right\} (1+o(1)).
\end{eqnarray}
where $o(1) \to 0$ as $\min\{\gamma, \frac{1}{\alpha}\} \to \infty$.
Thus, the asymptotically optimal upper bound on $\add$
(which corresponds to the optimum centralised procedure $\tau^*$) and that of the distributed
procedures $\mathsf{ALL}$, ${\sf{HALL}}$, and ${\sf{MAX}}$ scale in the
same way as
$\ln\gamma/\mathsf{KL}(f_1,f_0)$ and $-\ln\alpha/\mathsf{KL}(f_1,f_0)$. 

\section{Numerical Results}
\label{sec:numerical_results}
We consider a deployment of 7 nodes with the detection range $\DetRange
= 1$, in a hexagonal $\ROI$ (see Fig.~\ref{fig:sensor_placement_7nodes})
such that we get $N = 12$ detection subregions, and ${\cal N}_1 =
\{1,3,4,6\}$, ${\cal N}_2 = \{1,3,4\}$, ${\cal N}_3 = \{1,2,3,4\}$,
${\cal N}_4 = \{1,2,4\}$, ${\cal N}_5 = \{1,2,4,5\}$, ${\cal N}_6 =
\{2,4,5\}$, ${\cal N}_7 = \{2,4,5,7\}$, ${\cal N}_8 = \{4,5,7\}$, ${\cal
N}_9 = \{4,5,6,7\}$, ${\cal N}_{10} = \{4,6,7\}$, ${\cal N}_{11} =
\{3,4,6,7\}$, and ${\cal N}_{12} = \{3,4,6\}$. The pre--change pdf
considered is $f_0 \sim \mathcal{N}(0,1)$, and the detection range and the
influence range considered are $\DetRange=1.0$ and $\InfRange=1.5$ respectively.

We compute the $\add$, the $\tfa$ and the $\pfi$ performance of
$\sf{MAX}$, $\sf{HALL}$, $\sf{ALL}$, and Nikiforov's procedure
(\cite{nikiforov03lower-bound-for-det-isolation}) for the Boolean
sensing model with $f_1 \sim \mathcal{N}(1,1)$, and plot the $\add$ vs
$\log(\tfa)$ performance in Fig.~\ref{fig:bool}, of the change
detection/isolation procedures for $\pfi \leq 5 \times 10^{-2}$.  The
local $\CUSUM$ threshold $c$ that yields the target $\tfa$ and other
simulation parameters and results  are tabulated in
Table~\ref{tab:bool}.  To obtain the $\add$ the event is assumed to
occur at time 1, which corresponds to the 
maximum mean delay (see \cite{pollak85},
\cite{stat-sig-proc.lorden71procedures-change-distribution}). 
We observe from Fig.~\ref{fig:bool} that the $\add$ performance of
$\mymax$ is the worst and that of ${\sf Nikiforov}'s$ is the best.
Also, we note that the performance of the distributed procedures, 
${\sf ALL}$ and ${\sf HALL}$, are very close to that of the optimal 
centralised procedure. For eg., for a
requirement of $\tfa = 10^5$ (and $\pfi \leq 5 \times 10^{-2}$), we 
observe from Fig.~\ref{fig:bool} that $\add(\tau^{\mymax}) = 26.43$, 
$\add(\tau^{\myhall}) = 13.78$, 
$\add(\tau^{\myall}) = 12.20$, and 
$\add(\tau^{*}) = 11.28$. Since $\mymax$ does not make use of the
the dynamics of $C_k^{(s)}$ beyond $\tau^{s}$, it's $\add$ vs $\tfa$
performance is poor. On the other hand, $\myall$ and $\myhall$     
make use of $C_k^{(s)}$ for all $k$ and hence, give a better
performance.


\begin{figure}
\centerline{\includegraphics[width = 55mm, height =50mm]{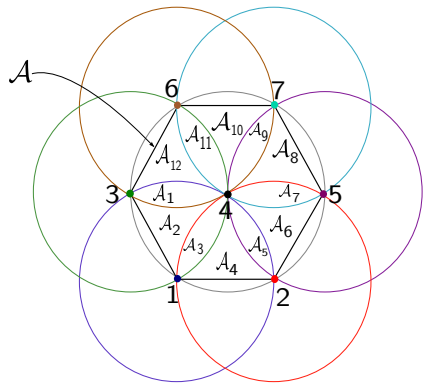}}
\caption{{\bf Sensor nodes placement}: 7 sensor nodes (which are
numbered 1,2,$\cdots$,7) represented by small filled circles are placed
in the hexagonal $\ROI$ ${\cal A}$. 
The sensor nodes partition the $\ROI$ into the detection subregions 
${\cal A}_1, {\cal A}_2, \cdots, {\cal A}_{12}$ (for both the Boolean
and the power law path loss sensing models). }
\label{fig:sensor_placement_7nodes}
\end{figure}

\begin{figure}[t]
\centering
\subfigure[$\add$ vs $\tfa$ for the Boolean model]
{ 
\includegraphics[width=73mm,height=50mm]{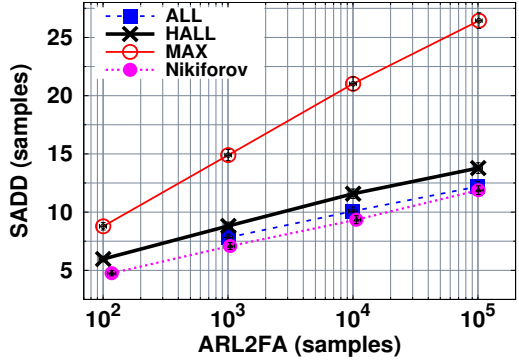}
\label{fig:bool}
}
	\hspace{10mm}
\subfigure[$\add$ vs $\tfa$ for the square law path loss model]
{ 
\includegraphics[width=73mm,height=50mm]{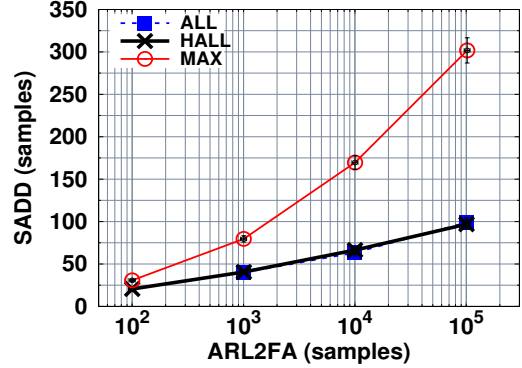}
\label{fig:path}
}
\caption{
          $\add$ versus $\tfa$ (for $\pfi \leq 5 \times 10^{-2}$) 
		  for $\sf{MAX}$, $\sf{HALL}$, $\sf{ALL}$ and
          Nikiforov's procedure for the Boolean and the square law path
		  loss sensing models. In the Boolean sensing model, the  
          system parameters are $f_0 \sim {N}(0,1)$, $f_1 \sim
		  {N}(0,1)$, and in the case of path loss sensing model, the
		  parameters are 
          $f_0 \sim {N}(0,1)$, $h_e = 1$, $\DetRange = 1.0$, $\InfRange=1.5$.}
\label{fig:sadd_vs_tfa}
\end{figure}

\begin{table} \footnotesize
\begin{center}
\caption{Simulation parameters and results 
for the Boolean sensing model for $\pfi \leq 5 \times 10^{-2}$}
\begin{tabular} {|c|l|l|c|r|r|r|r|r|}
\hline
Detection/ &  No. of & Threshold &  &   \multicolumn{2}{|c|}{99\% Confidence interval}&  &   \multicolumn{2}{|c|}{99\% Confidence interval}   \\\cline{5-6}\cline{8-9} 
 Isolation &  MC     &  $c$      &  $\tfa$        &  $\tfa_{\text{lower}}$ & $\tfa_{\text{upper}}$&   $\add$       &  $\add_{\text{lower}}$ & $\add_{\text{upper}}$ \\ 
 procedure &  runs   &           &                   &                          &                        &   &  &   \\ \hline
\multirow{5}{*}{{\sf MAX}}
& $10^4$ & 2.71 &  $10^2$ &    93.69 &   106.61 &     8.77	& 8.45	&  9.09\\ \cline{2-9}
& $10^4$ & 4.93 &  $10^3$ &   942.10 &  1065.81 &    14.89	& 14.41	& 15.37\\ \cline{2-9}
& $10^4$ & 7.24 &  $10^4$ &  9398.61 & 10640.99 &    21.01	& 20.42	& 21.61\\ \cline{2-9}
& $10^4$ & 9.52 &  $10^5$ & 95696.90 & 108008.89&    26.43	& 25.76	& 27.11\\ \hline
\multirow{4}{*}{{\sf HALL}}
& $10^4$ & 1.67 &  $10^2$ & 92.67       & 107.58   &  5.96 &	 5.72 &	 6.20\\ \cline{2-9}  
& $10^4$ & 2.69 &  $10^3$ & 927.17      & 1085.48  &  8.81 &	 8.48 &	 9.14\\ \cline{2-9}
& $10^4$ & 3.66 &  $10^4$ & 9239.97     & 10826.71 & 11.58 &	11.17 &	11.99\\ \cline{2-9}
& $10^4$ & 4.52 &  $10^5$ & 92492.85    & 108389.15& 13.78 &	13.32 &	14.23\\ \hline   
\multirow{4}{*}{{\sf ALL}}
& $10^4$ & 2.16 &  $10^3$ &   915.94  &   1089.33 &   7.82 &	 7.53 &	  8.11\\ \cline{2-9}  
& $10^4$ & 2.96 &  $10^4$ &  9197.23  &  10811.90 &  10.07 &	 9.70 &	 10.44\\ \cline{2-9} 
& $10^4$ & 3.71 &  $10^5$ &  92205.45 & 107952.43 &  12.20 &	11.76 &	 12.63\\ \hline 
\multirow{4}{*}{{\sf Nikiforov}}
& $10^4$ & 2.75   & $10^2$ &   98.30  &    116.32 &   4.75	   &   4.52   &	 4.98\\ \cline{2-9}
& $10^4$ & 4.50   & $10^3$ &  986.48  &   1048.23 &   7.08	   &   6.79	  &  7.38\\ \cline{2-9}
& $10^4$ & 6.32   & $10^4$ & 9727.19  &  10261.94 &    9.14   &   9.00   &   9.68\\ \cline{2-9}
& $10^4$ & 8.32   & $10^5$ & 98961.41 & 110415.50 &  11.28	  &  11.00  & 12.25\\
\hline
\end{tabular}
\label{tab:bool}
\end{center}
\end{table}

\begin{table}
\begin{center}
\caption{Simulation parameters and results 
for the square law path loss sensing model for $\pfi \leq 5 \times
	10^{-2}$}
\begin{tabular} {|c|l|l|c|r|r|r|r|r|}
\hline
Detection/ &  No. of & Threshold &  &   \multicolumn{2}{|c|}{99\% Confidence interval}&  &   \multicolumn{2}{|c|}{99\% Confidence interval}   \\\cline{5-6}\cline{8-9} 
 Isolation &  MC     &  $c$      &  $\tfa$        &  $\tfa_{\text{lower}}$ & $\tfa_{\text{upper}}$&   $\add$       &  $\add_{\text{lower}}$ & $\add_{\text{upper}}$ \\ 
 procedure &  runs   &           &                   &                          &                        &   &  &   \\ \hline
\multirow{5}{*}{{\sf MAX}}
& $10^4$ & 2.71 &  $10^2$ &    93.69 &    106.61 &   30.74  &	 29.31 &  32.17\\ \cline{2-9}
& $10^4$ & 4.93 &  $10^3$ &   942.10 &   1065.81 &   79.60  &	 75.86 &  83.34\\ \cline{2-9}
& $10^4$ & 7.23 &  $10^4$ &  9398.61 &  10640.99 &  169.63  &	161.61 & 177.65\\ \cline{2-9}
& $10^4$ & 9.52 &  $10^5$ & 95696.90 & 108008.89 &  301.77  &	286.88 & 316.66\\ \hline
\multirow{4}{*}{{\sf HALL}}
& $10^4$ & 1.67 &   $10^2$ &    92.67 &    107.58 & 20.58  &	19.43 &	21.74\\ \cline{2-9}
& $10^4$ & 2.69 &   $10^3$ &   927.17 &   1085.48 & 40.56  &	38.24 &	42.88\\ \cline{2-9}
& $10^4$ & 3.66 &   $10^4$ &  9239.97 &  10826.71 & 66.45  &	62.57 &	70.33\\ \cline{2-9}
& $10^4$ & 4.52 &   $10^5$ & 92492.85 & 108389.15 & 96.93  &	91.03 & 102.82\\ \hline 
\multirow{4}{*}{{\sf ALL}}
& $10^4$ & 1.33 &  $10^2$ &    92.24  &    107.79 &   20.19 &	19.06&21.32\\ \cline{2-9}  
& $10^4$ & 2.16 &  $10^3$ &   915.94  &   1089.33 &   39.90 &   37.59&42.21\\ \cline{2-9}  
& $10^4$ & 2.96 &  $10^4$ &  9197.23  &  10811.90 &   63.34 &   59.43&67.24\\ \cline{2-9} 
& $10^4$ & 3.71 &  $10^5$ &  92205.45 & 107952.43 &   98.96 &   93.01&104.92\\ \hline 
\end{tabular}
\label{tab:path}
\end{center}
\end{table}

\normalsize

For the same sensor deployment in
Fig.~\ref{fig:sensor_placement_7nodes}, we compute the $\add$ and the
$\tfa$ for the square law path loss ($\eta=2$) sensing model given in
Section~\ref{sec:system_model}. Also, 
the signal strength $h_e$ is taken to be unity. Thus, the sensor sets
(${\cal N}_i$s) and the detection subregions (${\cal A}_i$s) are the same as
in the Boolean model, we described above.
Since $\DetRange$
is taken as 1, $f_1(\cdot;\DetRange) \sim {\cal N}(1,1)$.
Thus, the LLR of observation $X_k^{(s)}$ is given by 
$\ln\left(\frac{f_1(X_k^{(s)};\DetRange)}{f_0(X_k^{(s)})}\right) =
X_k^{(s)}-\frac{1}{2}$, which is the same
as that in the Boolean sensing model. Hence, under {\em the event not occurred
hypothesis}, the $\tfa$ under the path
loss sensing model is the same as that of the Boolean sensing model. 
The $\CUSUM$ threshold $c$ that yields the target $\tfa$s and
other parameters and results  
are tabulated in Table~\ref{tab:path}.
To obtain the $\add$ the event is assumed to
occur at time 1, and at a distance of
$\InfRange$ from all the nodes of ${\cal N}_i$ that influence covers the
event (which corresponds to the maximum detection delay).
We plot the $\add$ vs $\log(\tfa)$ in
Fig.~\ref{fig:path}. The ordering on $\add$ for any $\tfa$ 
across the procedures is the same as that
in the Boolean model, and can be explained in the same manner. The
ambiguity in $\ell_e$ affects $f_1(\cdot;d_{e,s})$ and shows up as 
large $\add$ values.


\vspace{-4mm}

\section{Conclusion}
\label{sec:conclusions}
We consider the quickest distributed event detection/isolation problem
in a large extent $\wsn$ with a practical sensing model which
incorporates the reduction in signal strength with distance. We formulate the change
detection/isolation problem in the optimality framework of  
\cite{nikiforov03lower-bound-for-det-isolation}
and \cite{tartakovsky08multi-decision}.
We propose distributed
detection/isolation procedures, $\sf{MAX}$, $\sf{ALL}$ and $\sf{HALL}$
and show that as $\min\{\tfa,1/\pfi\} \to \infty$, the $\add$ performance
of the distributed procedures grows in the same scale as that of the  
optimal centralised procedure of Tartakovsky  
\cite{tartakovsky08multi-decision} and Nikiforov 
\cite{nikiforov03lower-bound-for-det-isolation}.

\vspace{-4mm}

\appendices
\label{sec:appendix}

\section{Proof of Theorem~\ref{thm:tfa}}
\label{app:tfa}
From detection sensor sets ${\cal N}_i, i=1,2,\cdots,N$, we choose the
collection of indices ${\cal I} \subseteq \{1,2,\cdots,N\}$ such that
any two sensor sets ${\cal N}_i$, ${\cal N}_j$, $i,j
\in {\cal I}$, are not partially ordered by set inclusion. For each $i 
\in {\cal I}$, define the set of sensors that are unique to the
sensor set ${\cal N}_i$, 
${\cal M}_i  \ := \ {\cal N}_i\setminus \underset{j \neq i, j \in {\cal
I}}{\bigcup}{\cal N}_j \  
\subseteq {\cal N}_i$.
The sets ${\cal M}_1, {\cal M}_2, \cdots,
{\cal M}_{|{\cal I}|}$ are disjoint. Under the null hypothesis, ${\bf H}_0$, the
observations of sensors in the sensor sets ${\cal M}_1, {\cal M}_2,
\cdots, {\cal M}_{|{\cal I}|}$ are iid, with
the pdf 
$f_0 \sim {\cal N}(0,\sigma^2)$. 
For every ${\cal N}_i$, there exists ${\cal M}_j$ such that  
${\cal M}_j \subseteq {\cal N}_i$, so that 
$\tau^{{\sf rule},({\cal N}_i)} \geq \tau^{{\sf rule},({\cal M}_j)}$.
Hence,
$\tau^{\sf procedure}  =\min\{\tau^{{\sf procedure},({\cal N}_i)}:i=1,2,\cdots,N\} 
\geq \min\{\tau^{{\sf procedure},({\cal M}_i)}:i\in {\cal I}\} 
=: \widehat{\tau}^{\ {\sf rule}}$.
Hence,
\begin{align}
\myexpnull{\tau^{{\sf rule}}}
&\geq    \myexpnull{\widehat{\tau}^{\ {\sf rule}}} 
\ \geq e^{mc} \cdot  \prob{\widehat{\tau}^{\ {\sf rule}} > e^{mc}}
\  \text{(by the Markov inequality)}\nn
\text{or,} \
\frac{\myexpnull{\tau^{{\sf rule}}}}{e^{mc}}
&\geq   \prob{\widehat{\tau}^{\ {\sf rule}} > e^{mc}} 
\label{eqn:markov}
\ =  \ \prod_{i \in {\cal I}}\probnull{\tau^{\ {\sf rule}, ({\cal M}_i)} > e^{mc}}.
\end{align}
We analyse 
$\probnull{\tau^{\ {\sf rule}, ({\cal M}_i)} > e^{mc}}$ as $c \to
\infty$, for 
$\myall$, $\mymax$, and $\myhall$. For ${\sf ALL}$,
\begin{align}
\probnull{\tau^{\ {\sf ALL}, ({\cal M}_i)} = k}
& \leq \ \probnull{C_k^{(s)} \geq c, \forall s \in {\cal M}_i }\
\ = \ \prod_{s \in {\cal M}_i } \probnull{C_k^{(s)} \geq c }\nn
& {\leq} \ e^{-cm_i} \ \ \ \ \text{(using Wald's inequality)}\nn 
\text{Therefore}, \hspace{10mm} 
\probnull{\tau^{\ {\sf ALL}, ({\cal M}_i)} \leq k}
&\leq k \cdot e^{-cm_i}\nn
\probnull{\tau^{\ {\sf ALL}, ({\cal M}_i)} > e^{mc}} &\geq 1 -
e^{-c(m_i-m)}. \nonumber
\end{align}
Hence, for any $m < m_i$, we have
$\liminf_{c \to \infty} \ 
\probnull{\tau^{\ {\sf ALL}, ({\cal M}_i)} > e^{mc}} = 1$.
A large $m$ (which is smaller than all $m_i$s) is desirable. Thus, a
good choice for $m$ is
$a_{\sf ALL} = \min\{m_i : i \in {\cal I}\} - \delta$.
for some arbitrarily small $\delta > 0$.
Hence, from Eqn.~\eqref{eqn:markov},
\begin{align}
\myexpnull{\tau^{\sf ALL}}& \geqslant \exp\left( a_{\sf ALL}c \right)(1+o(1)) 
\end{align}

\noindent
For ${\sf MAX}$,
at the stopping time of $\mymax$, at least one of the ${\sf CUSUM}$
statistics is above the threshold $c$,
\begin{align}
\probnull{\tau^{\ {\sf MAX}, ({\cal M}_i)} = k}
&\leq  \probnull{C_k^{(s)} \geq c, \ \text{for some} \ s \in {\cal M}_i }\nonumber \\
&\leq  \sum_{s \in {\cal M}_i} \probnull{C_k^{(s)} \geq c }\nonumber \\
&{\leq}  m_i e^{-c} \ \ \ \ \text{(using Wald's inequality)}. 
\end{align}
\begin{align}
\text{Therefore, for any arbitrarily small $\delta > 0$}, \
\probnull{\tau^{\ {\sf MAX}, ({\cal M}_i)} > e^{(1-\delta)c}} &\geq 1- m_ie^{-\delta c} \nonumber \\ 
\liminf_{c \to \infty}
\probnull{\tau^{\ {\sf MAX}, ({\cal M}_i)} > e^{(1-\delta)c}} &=1. 
\end{align}
Let $a_{\sf MAX} = 1-\delta$. 
For any arbitrarily small $\delta >0$, we see from Eqn.~\eqref{eqn:markov},
\begin{align}
\label{eqn:max_arlfa}
\myexpnull{\tau^{{\sf MAX}}}
\geqslant \exp\left((1-\delta) c\right) (1+o(1)) \
\ =: \ \exp\left(a_{\sf MAX} c\right) (1+o(1)),
\end{align}

\noindent
For {\sf HALL}, for the same threshold $c$, 
the stopping time of $\myhall$ is after
that of $\mymax$. Hence, 
$\tau^{\ {\sf HALL}} \ \geq \ \tau^{\ {\sf MAX}}$. Hence,
$\myexpnull{\tau^{\ {\sf HALL}}} \ \geq \ 
\myexpnull{\tau^{\ {\sf MAX}}} \ \geq \ 
\exp\left((1-\delta)c\right)(1+o(1))$ (from Eqn.~\eqref{eqn:max_arlfa}). 
Thus, for $a_{\sf ALL} := 1-\delta$, for any arbitrarily small $\delta >
0$,
\begin{align}
\myexpnull{\tau^{\ {\sf HALL}}} & \geq \exp\left(a_{\sf ALL} c\right) (1+o(1))
\end{align}
\section{Proof of Theorem~\ref{thm:pfi}}
\label{app:pfi}
Consider $\ell_e \in {\cal A}_i$. The probability of false isolation
when the detection is due to ${\cal N}_j \not\subseteq{\cal N}(\ell_e)$
is
\begin{align*}
\pmeasure{	
	\tau^{{\sf rule}} = \tau^{{\sf rule},({\cal N}_j)} 
	     }
&= \pmeasure{	
	\tau^{{\sf rule},({\cal N}_j)} \leq 
		\tau^{{\sf rule},({\cal N}_h)}, \forall h=1,2,\cdots,N 
	}\\
&\leq \pmeasure{	
	\tau^{{\sf rule},({\cal N}_j)} \leq \tau^{{\sf rule},({\cal N}_i)}
	} \\
&= \sum_{k=1}^\infty 
    \pmeasure 
	{ \tau^{{\sf rule},({\cal N}_i)} = k
	} 
    \pmeasure{ \tau^{{\sf rule},({\cal N}_j)} \leq k  
	\mid  \tau^{{\sf rule},({\cal N}_i)} = k
	}\nn
&= \sum_{k=1}^\infty 
    \pmeasure 
	{ \tau^{{\sf rule},({\cal N}_i)} = k
	} 
    \left[    
 \sum_{t=1}^k
	\pmeasure{ \tau^{{\sf rule},({\cal N}_j)} = t  
	\mid  \tau^{{\sf rule},({\cal N}_i)} = k
	}
	\right]
\end{align*}
\subsection{${\sf PFI}(\tau^{\sf ALL})$ -- Boolean Sensing Model}

\vspace{-9mm}

\begin{align}
\pbmeasure{
	\tau^{\myall,({\cal N}_j)}=t 
	\mid \tau^{\myall,({\cal N}_i)}=k
	}
&\leq  \pbmeasure{ C_{t}^{(s)} \geq c, \ \forall s \in {\cal N}_j 
	 \mid \tau^{\myall,({\cal N}_i)}=k
	 }\nn
&\leq  {\mathsf P}_\infty\left\{ C_{t}^{(s)} \geq c, \
\forall s \in {\cal N}_j \setminus {\cal N}_i 
	 \right\}\nn
&\leq
 \exp\left(-|{\cal N}_j\setminus{\cal N}_i|c\right) \ \ \
\text{(using Wald's inequality)}.\nn 
\text{Therefore,} \
\pbmeasure{
	\tau^{{\myall},({\cal N}_j)} \leq 
	 \tau^{\myall,({\cal N}_i)} 
	 }
&\leq   \exp\left(-|{\cal N}_j\setminus{\cal N}_i|c\right) \cdot 
{\mathsf E}_1^{(i)}\left[\tau^{\myall, ({\cal N}_i)} \right] \nonumber \\
&\leq   \exp\left(-(|{\cal N}_j\setminus{\cal N}_i|c-\ln(c))\right) \cdot
\frac{1}{\alpha|{\cal N}_i|}(1+o(1)). \nn
\text{Hence}, \ {\sf PFI}(\tau^{\myall}) 
&\leq \ \max_{1 \leq i \leq N} \ 
\max_{1 \leq j \leq N, {\cal N}_j \not\subseteq {\cal N}_i } 
\ 
{\mathsf P}_1^{(i)}\left\{
	\tau^{{\myall},({\cal N}_j)} \leq \tau^{\myall,({\cal N}_i)}
	\right\}\nn
&\leq    \frac{\exp\left(-(m c-\ln(c))\right)}{\underline{n}\alpha
}(1+o(1))
\end{align}
where 
$\underline{n} = \min\{|{\cal N}_i:i=1,2,\cdots,N|\}$, 
$m = \underset{1 \leq i \leq N, 1 \leq j \leq N, {\cal N}_j
\not\subseteq {\cal N}_i}{\min} \{ |{\cal
N}_j\setminus{\cal N}_i| \}$.
For any $n$, there exists $c_0(n)$ such that for all $c
> c_0(n), c < e^{c/n}$. Using this inequality, for sufficiently large $c$
\begin{align*}
{\sf PFI}(\tau^{\myall}) 
&\leq    \frac{\exp\left(-\left(\left(m-\frac{1}{n}\right) c\right)\right)}{\underline{n}\alpha
}(1+o(1)) \ = 
\frac{\exp(- b_{\myall}\cdot c)}
 {B_{\sf \myall}}
 (1+o(1)),
\end{align*}
where
$b_{\sf \myall} = m - 1/n$ and $B_{\myall} = \underline{n}\alpha$. 
\subsection{${\sf PFI}(\tau^{\sf MAX})$ -- Boolean Sensing Model}

\vspace{-9mm}

\begin{align*}
{\mathsf P}_1^{(i)}\left\{ \tau^{{\sf MAX},({\cal N}_j)} =t
\mid \tau^{\mymax, ({\cal N}_i)}=k 
\right\}
&\leq
{\mathsf P}_1^{(i)}\left\{ \tau^{(s)} \leq t, \forall s \in {\cal
N}_j
\mid \tau^{\mymax, ({\cal N}_i)}=k 
\right\}\\
&\leq
{\mathsf P}_\infty\left\{ \tau^{(s)} \leq t, \forall s \in {\cal
N}_j\setminus{\cal N}_i
\mid \tau^{\mymax, ({\cal N}_i)}=k 
\right\}\\
&=
\prod_{s \in {\cal N}_j\setminus {\cal N}_i} \
\sum_{n=1}^t \ 
{\mathsf P}_\infty\left\{ \tau^{(s)} =n \mid \tau^{\mymax, ({\cal
N}_i)}=k 
\right\}\\
&=
\prod_{s \in {\cal N}_j\setminus {\cal N}_i} \
\sum_{n=1}^t \ 
{\mathsf P}_\infty\left\{ C_n^{(s)} \geq c 
\right\}\\
&\leq \exp\left(-m_{ji}c\right)t^{m_{ji}}\\ 
\text{Hence,} \ {\mathsf P}_1^{(i)}\left\{
\tau^{\sf MAX,({\cal N}_j)} \leq \tau^{\mymax, ({\cal N}_i)} 
\right\}
&\leq 
\exp\left(- m_{ji} c\right)
{\mathsf E}_1^{(i)}\left[(\tau^{\mymax, {\cal N}_i})^{1+m_{ji}}
\right] \nonumber\\ 
&\leq  
\exp\left(- m_{ji} c\right)  
\frac{c^{1+m_{ji}}}{\alpha^{1+m_{ji}}}(1+o(1))\\
& =  
\frac{\exp\left(-\left(m_{ji}c-(1+m_{ji})\ln(c)\right)\right)}{\alpha^{1+m_{ji}}}
(1+o(1))
\end{align*}
Let $m = \underset{1 \leq i \leq N, 1 \leq j \leq N, {\cal N}_j
\not\subseteq {\cal N}_i}{\min} m_{ji}$,
  $\bar{m} = \underset{1 \leq i \leq N, 1 \leq j \leq N, {\cal N}_j
\not\subseteq {\cal N}_i}{\max} m_{ji}$, 
and  $\alpha^* = \underset{1 \leq i \leq N, 1 \leq j \leq N, {\cal N}_j
\not\subseteq {\cal N}_i}{\min} \alpha^{1+m_{ji}}$.
\begin{align*}
\text{Therefore,} \ 
{\sf PFI}(\tau^{\mymax}) 
&\leq  \max_{1 \leq i \leq N} \ 
\max_{1 \leq j \leq N, {\cal N}_j \not\subseteq {\cal N}_i } 
\ 
{\mathsf P}_1^{(i)}\left\{
\tau^{\sf MAX,({\cal N}_j)} \leq \tau^{\mymax, ({\cal N}_i)} 
\right\}\\
&\leq 
\frac{\exp\left(-\left(m c-(1+\bar{m})\ln(c)\right)\right)}{\alpha^*}
(1+o(1)).
\end{align*}
For any $n$, there exists $c_0(n)$ such that for all $c
> c_0(n), c < e^{c/n}$. Hence, for sufficiently large $c$
\begin{align*}
{\sf PFI}(\tau^{\mymax}) 
&\leq  \max_{1 \leq i \leq N} \ 
\max_{1 \leq j \leq N, {\cal N}_j \not\subseteq {\cal N}_i } 
\ 
{\mathsf P}_T^{(i)}\left\{
\tau^{\sf MAX,({\cal N}_j)} \leq \tau^{\mymax, ({\cal N}_i)} 
\right\}\\
&\leq 
\frac{\exp\left(-\left(m -\frac{1+\bar{m}}{n}\right)c\right)}{\alpha^*}
(1+o(1)) \ = \ 
 \frac{\exp(- b_{\mymax} \cdot c)}
 {B_{\sf \mymax}}
 (1+o(1)), 
\end{align*}
where $b_{\sf \mymax} = m - ((1+\bar{m})/n)$ and $B_{\mymax} = \alpha^*$. 
\subsection{${\sf PFI}(\tau^{\sf HALL})$ -- Boolean Sensing Model}

\vspace{-9mm}

\begin{align*}
{\mathsf P}_1^{(i)}\left\{ \tau^{{\sf HALL},({\cal N}_j)} =t
\mid \tau^{\myhall, ({\cal N}_i)}=k 
\right\}
&\leq
{\mathsf P}_1^{(i)}\left\{ \tau^{(s)} \leq t, \forall s \in {\cal
N}_j
\mid \tau^{\myhall, ({\cal N}_i)}=k 
\right\}
\end{align*}
which has the same form as that of $\mymax$. Hence, from the analysis of
$\mymax$, it follows that
\begin{align*}
{\mathsf P}_1^{(i)}\left\{
\tau^{\sf HALL,({\cal N}_j)} \leq \tau^{\myhall, ({\cal N}_i)} 
\right\}
&\leq  
\exp\left(-m_{ji}c\right)
 {\mathsf E}_1^{(i)}\left[(\tau^{\myhall, ({\cal N}_i)})^{1+m_{ji}}
 \right]\\
&\leq  
\exp\left(-m_{ji}c\right)
\frac{c^{1+m_{ji}}}{|{\cal N}_i|^{1+m_{ji}} \alpha^{1+m_{ji}}}(1+o(1))\\
& =  
\exp\left(-\left(m_{ji}c-(1+m_{ji})\ln(c)\right)\right)
\left[\frac{1}{\alpha|{\cal N}_i| }\right]^{1+m_{ji}}
(1+o(1))
\end{align*}

\begin{align*}
{\sf PFI}(\tau^{\myhall})
& \leq \ \max_{1 \leq i \leq N} \ 
\max_{1 \leq j \leq N, {\cal N}_j \not\subseteq {\cal N}_i } 
\ 
{\mathsf P}_1^{(i)}\left\{
\tau^{\sf HALL,({\cal N}_j)} \leq \tau^{\myhall, ({\cal N}_i)}
\right\} \\
&\leq    
\frac{ \exp\left(-
\left(mc-(1+\bar{m})\ln(c)\right)\right)}{\alpha^*}(1+o(1)).
\end{align*}
where
$\alpha^* = \underset{1 \leq i \leq N, 1 \leq j \leq N, {\cal N}_j
\not\subseteq {\cal N}_i}{\min} \left(\alpha \cdot|{\cal
N}_i|\right)^{1+m_{ji}}$.
For any $n$ there exists $c_0(n)$ such that for all $c
> c_0(n), c < e^{c/n}$. Hence, for sufficiently large $c$
\begin{align*}
{\sf PFI}(\tau^{\myhall})
& \leq \ \max_{1 \leq i \leq N} \ 
\max_{1 \leq j \leq N, {\cal N}_j \not\subseteq {\cal N}_i } 
\ 
{\mathsf P}_1^{(i)}\left\{
\tau^{\sf HALL,({\cal N}_j)} \leq \tau^{\myhall, ({\cal N}_i)}
\right\} \\
&\leq    
\frac{ \exp\left(-
\left(m-\frac{1+\bar{m}}{n}\right)c\right)}{\alpha^*}(1+o(1))
\ = \
 \frac{\exp(- b_{\myhall} \cdot c)}
 {B_{\sf \myhall}} 
 (1+o(1)), 
\end{align*}
where $b_{\sf \myhall} = m - ((1+\bar{m})/n)$ and $B_{\myhall} = \alpha^*$.

\section*{$\pfi$ -- Path--Loss Sensing Model}

\begin{lemma}
\label{lem:path-bound}
For $s \in {\cal N}_j\setminus {\cal M}_e$ and for $t \geq T$, (with the
pre--change pdf
$f_0 \sim {\cal N}(0,\sigma^2)$ and the post--change pdf $f_1 \sim {\cal
N}(h_e\rho(r_s),\sigma^2)$)
\begin{align*}
 \pmeasure{ C_{t}^{(s)} \geq c }
 &\leq 
	 \exp\left(  -\frac{\underline{\omega}_0}{2}c\right)
	 \cdot
 \frac{\exp\left( -
 \frac{\alpha\underline{\omega}_0^2}{4}\right)}{1-\exp\left(-\frac{\alpha
 \underline{\omega}_0^2}{4}\right)},
\end{align*}
where we recall that the parameter $\underline{\omega}_0$ defines the
influence range, and $\alpha = $ KL$(f_1,f_0)$.
\end{lemma}
\noindent
{\em Proof:}
For $s \in {\cal N}_j\setminus {\cal M}_e$ and for $t \geq T$,
\footnotesize
\begin{align*}
&\pmeasure{ C_{t}^{(s)} \geq c }\\
&= \pmeasure{ \max_{1\leq n \leq t}
\sum_{k=1}^{n}
\ln\left(\frac{f_1(X^{(s)}_k;r_s)}{f_0(X^{(s)}_k)}\right) \geq c  }\\
&\leq \sum_{n=1}^{\infty} \pmeasure{ \sum_{k=1}^n  \ln\left(\frac{f_1(X^{(s)}_k;r_s)}{f_0(X^{(s)}_k)}\right) \geq c  }\\
&= \sum_{n=1}^{T-1} {\mathsf P}_{\infty}\left\{ \sum_{k=1}^n  \ln\left(\frac{f_1(X^{(s)}_k;r_s)}{f_0(X^{(s)}_k)}\right) \geq c  \right\}
+ \sum_{n=T}^{\infty} \pmeasure{ \sum_{k=1}^n  \ln\left(\frac{f_1(X^{(s)}_k;r_s)}{f_0(X^{(s)}_k)}\right) \geq c  }\\
&= \sum_{n=1}^{T-1} {\mathsf P}_{\infty}\left\{ \sum_{k=1}^n  \ln\left(\frac{f_1(X^{(s)}_k;r_s)}{f_0(X^{(s)}_k)}\right) \geq c  \right\}
+
 \sum_{n=T}^{\infty} \pmeasure{
	\sum_{k=1}^{T-1}  \ln\left(\frac{f_1(X^{(s)}_k;r_s)}{f_0(X^{(s)}_k)}\right) 
	+\sum_{k=T}^{n}  \ln\left(\frac{f_1(X^{(s)}_k;r_s)}{f_0(X^{(s)}_k)}\right) \geq c  }\\
&= \sum_{n=1}^{T-1} {\mathsf P}_{\infty}\left\{
	\sum_{k=1}^{n}
	X_k^{(s)}
 \geq  \frac{\sigma^2}{h_e\rho(r_s)}(c + n \alpha)
\right\} 
+ \sum_{n=T}^{\infty} {\mathsf P}_{\infty}\left\{
	\sum_{k=1}^{n}
	X_k^{(s)}
	\geq \frac{\sigma^2}{h_e\rho(r_s)}c+nh_e\left(\frac{\rho(r_s)}{2}
 -    \rho(d_{e,s}) \right)
 +  \left( T-1\right) h_e\rho(d_{e,s})
	\right\}\\
&\leq \sum_{n=1}^{T-1} {\mathsf P}_{\infty}\left\{
	\sum_{k=1}^{n}
	X_k^{(s)}
 \geq  \frac{\sigma^2}{h_e\rho(r_s)}(c + n \alpha)
\right\} + \sum_{n=T}^{\infty} {\mathsf P}_{\infty}\left\{
	\sum_{k=1}^{n} X_k^{(s)} \geq 
	n \cdot {h_e \frac{\rho(r_s)}{2} \underline{\omega}_0}
    +  {c  \cdot  \frac{\sigma^2}{h_e\rho(r_s)}}
	  \right\}\\
&\leq \sum_{n=1}^{\infty} {\mathsf P}_{\infty}\left\{
	 \exp\left(\theta\sum_{k=1}^{n} X_k^{(s)}\right) \geq
	 \exp\left(
	 \frac{\theta\sigma^2}{h_e\rho(r_s)}(c+ n\alpha\underline{\omega}_0)
	 \right) 
	 \right\} \ \ \text{for any} \ \theta > 0.\\
\text{Hence}, \ \pmeasure{ C_{t}^{(s)} \geq c }
	 &\leq 
	 \sum_{n=1}^{\infty} 
	 \exp\left(-\frac{\theta\sigma^2}{h_e\rho(r_s)}
	 (c+n\alpha\underline{\omega}_0) \right) \left({\mathsf
	 E}_\infty\left[e^{\theta X_1^{(s)}}\right]\right)^n \\
	 &=
	 \sum_{n=1}^{\infty} 
	 \exp\left(-\frac{\theta\sigma^2}{h_e\rho(r_s)}
	 (c+n\alpha\underline{\omega}_0) + 
	 \frac{n\sigma^2\theta^2}{2}\right)
\end{align*}
\normalsize
Since the above inequality holds for any $\theta > 0$, we have
\begin{align*}
 \pmeasure{ C_{t}^{(s)} \geq c }
	 &\leq 
	 \sum_{n=1}^{\infty} 
	 \min_{\theta>0}
	 \exp\left(-\frac{\theta\sigma^2}{h_e\rho(r_s)}
	 (c+n\alpha\underline{\omega}_0) + 
	 \frac{n\sigma^2\theta^2}{2}\right)
\end{align*}
The minimising  $\theta$ is
$\frac{c+n\alpha\underline{\omega}_0}{nh_e\rho(r_s)}$.
Therefore, for 
$\theta = \frac{c+n\alpha\underline{\omega}_0}{nh_e\rho(r_s)}$,
\begin{align*}
 \pmeasure{ C_{t}^{(s)} \geq c }
	 &\leq 
	 \sum_{n=1}^{\infty} 
	 \exp\left( \frac{-(c+n\alpha\underline{\omega}_0)^2}{4\alpha n}
	 \right). 
\end{align*}
\begin{align*}
\text{Note that} \ 
	 -\frac{(c+ \alpha \underline{\omega}_0 n)^2}{4\alpha n} 
	 +\frac{(c+ \alpha \underline{\omega}_0 (n-1))^2}{4\alpha (n-1)} 
	 &=
	 -\frac{\alpha\underline{\omega}_0^2}{4} + \frac{c^2}{4\alpha(n-1)n}
\end{align*}
Therefore, by iteratively computing the exponent, we have
\begin{align*}
	 \exp\left(-\frac{(c+\alpha\underline{\omega}_0n)^2}{4\alpha n}\right) 
	 &= 
	 \exp\left(-\frac{(c+\alpha\underline{\omega}_0)^2}{4\alpha}\right)\cdot 
	 \exp\left(-\frac{\alpha \underline{\omega}_0^2}{4}(n-1)\right) \exp\left(
	 \frac{c^2}{4\alpha}\left(1 - \frac{1}{n}\right)\right) \\
	 &\leq 
	 \exp\left(-\frac{(c+\alpha\underline{\omega}_0)^2}{4\alpha}\right)\cdot 
	 \exp\left(-\frac{\alpha \underline{\omega}_0^2}{4}(n-1)\right) \exp\left(
	 \frac{c^2}{4\alpha}\right) \\
	 \text{or} \ 
	 \sum_{n=1}^{\infty}
	 \exp\left(-\frac{(c+\alpha\underline{\omega}_0n)^2}{4\alpha
	 n}\right) &\leq 
	 \exp\left(  -\frac{\underline{\omega}_0}{2}c\right)
	 \cdot 
	 \frac{\exp\left( - \frac{\alpha\underline{\omega}_0^2}{4}\right)}{1-\exp\left(-\frac{\alpha	 \underline{\omega}_0^2}{4}\right)}\\ 
	 &=: \beta
\end{align*}
\normalsize
%

\subsection{${\sf PFI}(\tau^{\sf ALL})$ -- Path Loss Sensing Model}

\vspace{-9mm}

\begin{align*}
\pmeasure{
\tau^{\sf ALL,({\cal N}_j)} = t 
\mid \tau^{\myall, ({\cal N}_i)}=k
} 
&\leq  
\pmeasure{C_t^{(s)} \geq c, \forall s
\in {\cal N}_j 
\mid \tau^{\myall, ({\cal N}_i)}=k
} \nonumber\nn
&\leq  
\pmeasure{C_t^{(s)} \geq c, \forall s
\in {\cal N}_j\setminus {\cal N}(\ell_e) 
\mid \tau^{\myall, ({\cal N}_i)}=k
} \nonumber \nn
&= \prod_{s \in {\cal N}_j\setminus{\cal N}(\ell_e)}
\pmeasure{C_t^{(s)} \geq c
} \nonumber \nn
&\leq \beta^{|{\cal N}_j\setminus{\cal N}(\ell_e)|}
\ \ \ (\text{from Lemma~\ref{lem:path-bound}})\nn 
\text{Therefore}, \
\pmeasure{
\tau^{\sf ALL,({\cal N}_j)} \leq \tau^{\myall, ({\cal N}_i)}
}
	&\leq \beta^{{|{\cal N}_j\setminus{\cal N}(\ell_e)|}}
{\mathsf E}_1^{({\bf d}(\ell_{e}))}\left[\tau^{\myall, ({\cal N}_i)}
\right]
	\\ 
	&\leq \beta^{{|{\cal N}_j\setminus{\cal N}(\ell_e)|}}\frac{c}{\alpha|{\cal
	N}_i|}(1+o(1)) 
\end{align*}
Let $m = \underset{1 \leq i \leq N, \ell_e \in {\cal A}_i, 1 \leq j \leq N, {\cal N}_j
\not\subseteq {\cal N}(\ell_e)}{\min} |{\cal N}_j \setminus {\cal
N}(\ell_e)|$
and
$\underline{n}=\min\{|{\cal N}_i|:i=1,2,\cdots,N\}$.
Define $K = 
 \underset{1 \leq i \leq N, \ell_e \in {\cal A}_i, 1 \leq j \leq N, {\cal N}_j
\not\subseteq {\cal N}(\ell_e)}{\max} 
\left[\frac{\exp\left(-\frac{\alpha\underline{\omega}_0^2}{4} \right)}
{1-\exp\left(-\frac{\alpha\underline{\omega}_0^2}{4}\right)}\right]^
{ |{\cal N}_j \setminus {\cal N}(\ell_e)|} $. 
Therefore,
\begin{align*}
{\pfi}\left(\tau^{\sf ALL}\right)
&\leq \ \max_{1 \leq i \leq N} \ 
\sup_{\ell_e \in {\cal A}_i} \ 
\max_{1 \leq j \leq N, {\cal N}_j \not\subseteq{\cal N}(\ell_i)} \ 
\pmeasure{
\tau^{\sf ALL,({\cal N}_j)} \leq \tau^{\myall, ({\cal N}_i)} 
} \\ 
&\leq   \frac{K \exp\left(-
\left(\frac{m\underline{\omega}_0}{2}c-\ln(c)\right) \right)}{\alpha
\underline{n}} (1+o(1)).
\end{align*}
For any $n$ there exists $c_0(n)$ such that for all $c
> c_0(n), c < e^{c/n}$. Hence, for sufficiently large $c$
\begin{align*}
{\pfi}\left(\tau^{\sf ALL}\right)
&\leq   \frac{K \exp\left(-
\left(\frac{m\underline{\omega}_0}{2}-\frac{1}{n}\right)c \right)}{\alpha
\underline{n}} (1+o(1)) \ = \ 
 \frac{\exp(-b_{\myall,d} \cdot c)}{B_{\sf \myall,d}} (1+o(1)) 
\end{align*}
where $b_{\myall,d}=(m\underline{\omega}_0/2) -(1/n)$ and 
$B_{\myall,d}= \alpha\underline{n}/K$.

\subsection{${\sf PFI}(\tau^{\sf MAX})$ -- Path Loss Sensing Model}

\vspace{-9mm}

\begin{align*}
\pmeasure{\tau^{\sf MAX,({\cal N}_j)} = t 
\mid \tau^{\mymax, ({\cal N}_i)}=k
} 
&\leq 
\pmeasure{\tau^{(s)} \leq t,  \forall \ s \in {\cal N}_j \setminus{\cal N}(\ell_e) 
\mid \tau^{\mymax, ({\cal N}_i)}=k
} \nonumber \\
&=
\prod_{s \in {\cal N}_j\setminus{\cal N}(\ell_e)} \
\pmeasure{\tau^{(s)} \leq t
\mid \tau^{\mymax, ({\cal N}_i)}=k
} \nonumber \\
&\leq
\prod_{s \in {\cal N}_j\setminus{\cal N}(\ell_e)} \
\sum_{n=1}^t \
\pmeasure{C_n^{(s)} \geq c} \nonumber \\
&\leq 
{\beta}^{|{\cal N}_j\setminus{\cal N}(\ell_e)|}\cdot 
t^{|{\cal N}_j\setminus{\cal N}(\ell_e)|}  
\ \ \ (\text{from Lemma~\ref{lem:path-bound}})\nn
\pmeasure{\tau^{\sf MAX,({\cal N}_j)} \leq \tau^{\mymax, ({\cal N}_i)}}
&\leq 
{\beta}^{|{\cal N}_j\setminus{\cal N}(\ell_e)|}\cdot 
{\mathsf E}_1^{({\bf d}(\ell_{e}))}\left[(\tau^{\mymax,({\cal N}_i)})^
{1+|{\cal N}_j\setminus{\cal N}(\ell_e)|}  
	\right] \nonumber \\
&\leq 
{\beta}^{|{\cal N}_j\setminus{\cal N}(\ell_e)|}\cdot 
\frac{c^{1+
|{\cal N}_j\setminus{\cal N}(\ell_e)|
}
}{\alpha^{1+
|{\cal N}_j\setminus{\cal N}(\ell_e)|
}}(1+o(1))
\end{align*}
Let $m = \underset{1 \leq i \leq N, \ell_e \in {\cal A}_i, 1 \leq j \leq N, {\cal N}_j
\not\subseteq {\cal N}(\ell_e)}{\min} 
|{\cal N}_j\setminus{\cal N}(\ell_e)|
$,
$\bar{m} = \underset{1 \leq i \leq N, \ell_e \in {\cal A}_i, 1 \leq j \leq N, {\cal N}_j
\not\subseteq {\cal M}_e}{\max} 
|{\cal N}_j \setminus {\cal N}(\ell_e)| 
$, and
define $K = 
 \underset{1 \leq i \leq N, \ell_e \in {\cal A}_i, 1 \leq j \leq N, {\cal N}_j
\not\subseteq {\cal N}(\ell_e)}{\max} 
\left[\frac{\exp\left(-\frac{\alpha\underline{\omega}_0^2}{4} \right)}
{1-\exp\left(-\frac{\alpha\underline{\omega}_0^2}{4}\right)}\right]^
{ |{\cal N}_j \setminus {\cal N}(\ell_e)|} $. 
Therefore,
\begin{align*}
{\sf PFI}(\tau^{\mymax}) 
&\leq \ \max_{1 \leq i \leq N} \ 
\sup_{\ell_e \in {\cal A}_i} \ 
\max_{1 \leq j \leq N, {\cal N}_j \not\subseteq{\cal N}(\ell_e)} \ 
\pmeasure{
\tau^{\sf MAX,({\cal N}_j)} \leq \tau^{\mymax, ({\cal N}_i)}
} \nonumber\\
&\leq   \frac{K}{\alpha^*} \exp\left(-
\left(\frac{m\underline{\omega}_0}{2}c-(1+\bar{m})\ln(c)\right)
\right)(1+o(1)).
\end{align*}
where $\alpha^* = \underset{1 \leq i \leq N, \ell_e \in {\cal A}_i, 1  \leq j \leq N, {\cal N}_j
\not\subseteq {\cal N}(\ell_e)}{\min} \alpha^{1+
|{\cal N}_j \setminus {\cal N}(\ell_e)|} 
$.
For any $n$ there exists $c_0(n)$ such that for all $c
> c_0(n), c < e^{c/n}$. Hence, for sufficiently large $c$
\begin{align*}
{\sf PFI}(\tau^{\mymax}) 
&\leq  \frac{K}{\alpha^*} \exp\left(-
\left(\frac{m\underline{\omega}_0}{2}-\frac{1+\bar{m}}{n}\right)c
\right)(1+o(1)) \ = \
 \frac{\exp(-b_{\mymax,d} \cdot c)}{
 B_{\sf \mymax,d} 
	 } (1+o(1)) ,
\end{align*}
where 
$b_{\sf \mymax,d}=
(\frac{m\underline{\omega}_0}{2})-(\frac{1+\bar{m}}{n})$ and
$B_{\mymax,d} = \frac{\alpha^*}{K}$.

\subsection{${\sf PFI}(\tau^{\sf HALL})$ -- Path Loss Sensing Model}

\vspace{-9mm}

\begin{align*}
\pmeasure{\tau^{\sf HALL,({\cal N}_j)} = t 
\mid \tau^{\myhall, ({\cal N}_i)}=k 
}
&\leq
\pmeasure{\tau^{(s)} \leq t,  \forall \ s \in {\cal N}_j \setminus {\cal
N}(\ell_e) 
\mid \tau^{\myhall, ({\cal N}_i)}=k
} 
\end{align*}
which has the same form as that of ${\sf MAX}$. Hence, from the analysis
of ${\sf MAX}$, it follows that
\begin{align*}
\pmeasure{
	\tau^{\sf HALL,({\cal N}_j)} \leq \tau^{\myhall, ({\cal N}_i)} 
	}
&\leq 
\beta^{
|{\cal N}_j\setminus{\cal N}(\ell_e)|}
{\mathsf E}_1^{({\bf d}(\ell_e))}\left[(\tau^{\myhall,({\cal N}_i)})^{
 1+ |{\cal N}_j\setminus{\cal N}(\ell_e)|
	} 
	\right] \nonumber \\
&\leq 
\beta^{
|{\cal N}_j\setminus{\cal N}(\ell_e)|
}
\frac{c^{1+
|{\cal N}_j\setminus{\cal N}(\ell_e)|
}}{
	(\alpha|{\cal N}_i|)^{1+
|{\cal N}_j\setminus{\cal N}(\ell_e)|
	}}(1+o(1)) \nn
\text{Therefore,} \
{\sf PFI}(\tau^{\myhall})
&\leq \ \max_{1 \leq i \leq N} \ 
\sup_{\ell_e \in {\cal A}_i} \ 
\max_{1 \leq j \leq N, {\cal N}_j \not\subseteq{\cal N}(\ell_e)} \ 
\pmeasure{
	\tau^{\sf HALL,({\cal N}_j)} \leq \tau^{\myhall, ({\cal N}_i)}
} \nonumber\\
&\leq   \frac{K}{\alpha^*} \exp\left(-
\left(\frac{m\underline{\omega}_0}{2}c-(1+\bar{m})\ln(c)\right)
\right)(1+o(1)).
\end{align*}
\begin{align*}
\text{Therefore for large $c$}, \pfi
&\leq   \frac{K}{\alpha^*} \exp\left(-
\left(\frac{m\underline{\omega}_0}{2}-\frac{1+\bar{m}}{n}\right)c
\right)(1+o(1)) \ = \ 
 \frac{ \exp(-b_{\myhall,d} \cdot c)}
 {B_{\sf \myhall,d}}
 (1+o(1)), 
\end{align*}
where $\alpha^* = \underset{1 \leq i \leq N, \ell_e \in {\cal A}_i, 1 \leq j \leq N, {\cal N}_j
\not\subseteq {\cal N}(\ell_e)}{\min} \left(\alpha \cdot|{\cal
N}_i|\right)^{1+
|{\cal N}_j\setminus{\cal N}(\ell_e)|
}$,
$b_{\sf \myhall,d} =
(m\underline{\omega}_0/2)-(1+\bar{m})/n$, and
$B_{\myhall,d} = \alpha^*/K$. 

\vspace{-5mm}

\section{$\add$ for the Boolean and the Path loss Models}
\label{app:sadd}
Fix $i, 1 \leq i \leq N$.  For each change time $T \geq 1$, define
$\mathcal{F}_T = \sigma(X^{(s)}_k, s \in \mathcal{N}, 1 \leq k \leq
T),$ and for $\ell_e \in {\cal A}_i$, $\mathcal{F}^{(i)}_T = \sigma(X^{(s)}_k, s \in
\mathcal{N}_i, 1 \leq k \leq T)$. From
\cite{stat-sig-proc.mei05information-bounds} (Theorem 3, Eqn.\ (24)), 
\begin{eqnarray} \label{eqn:mei05-thm3-24}
  \mathrm{ess} \sup \mathsf{E}^{({\bf d}(\ell_e))}_T 
               \left( (\tau^{{\sf rule},({\cal N}_i)} - T)^+ | \mathcal{F}^{(i)}_{(T-1)} \right) \leq
               \frac{c}{I} (1 + o(1)), \ 
			    \text{as $c \to \infty$}, 
\end{eqnarray}
Define
$\mathcal{F}_{\{\tau^{{\sf rule},({\cal N}_i)} \geq T\}}$ as the $\sigma$-field generated
by the event $\{\tau^{{\sf rule},({\cal N}_i)} \geq T\}$, and similarly define the
$\sigma$-field $\mathcal{F}_{\{\tau^{\sf rule}  \geq T\}}.$ Evidently
 $ \mathcal{F}_{\{\tau^{{\sf rule},(i)}  \geq T\}} \subset \mathcal{F}^{(i)}_{(T-1)} 
    \ \ \mathrm{and} \ \ 
  \mathcal{F}_{\{\tau^{\sf rule}  \geq T\}} \subset \mathcal{F}_{(T-1)}$.
By iterated conditional expectation, 
\begin{eqnarray} 
\label{eqn:bound_by_esssup}
  \mathsf{E}^{({\bf d}(\ell_e))}_T \left( (\tau^{{\sf rule},({\cal N}_i)} - T)^+ |
  \mathcal{F}_{\{\tau^{\sf rule}  \geq T\}} \right) 
  &\leq&
  \mathrm{ess} \sup \mathsf{E}^{({\bf d}(\ell_e))}_T 
  \left( (\tau^{{\sf rule},({\cal N}_i)} - T)^+ | \mathcal{F}_{(T-1)} \right) 
\end{eqnarray}
We can further assert that
\begin{eqnarray*}
  \mathsf{E}^{({\bf d}(\ell_e))}_T 
  \left( (\tau^{{\sf rule},({\cal N}_i)} - T)^+ | \mathcal{F}_{(T-1)} \right) 
  \stackrel{\mathrm{a.s.}}{=}
  \mathsf{E}^{({\bf d}(\ell_e))}_T 
  \left( (\tau^{{\sf rule},({\cal N}_i)} - T)^+ | \mathcal{F}^{(i)}_{(T-1)} \right)
\end{eqnarray*}
Using this observation with Eqn.~\ref{eqn:bound_by_esssup} and
Eqn.~\ref{eqn:mei05-thm3-24}, we can write, as $c \to \infty$, 
\begin{eqnarray} \label{eqn:bound_derived_from_mei05}
   \mathsf{E}^{({\bf d}(\ell_e))}_T \left( (\tau^{{\sf rule},({\cal N}_i)} - T)^+ |
   \mathcal{F}_{\{\tau^{\sf rule}  \geq T\}} \right)
   \leq \frac{c}{I} (1 + o(1))
\end{eqnarray}
 Finally, 
$  \mathsf{E}^{({\bf d}(\ell_e))}_T \left( (\tau^{{\sf rule},({\cal N}_i)} - T)^+ | \tau^{\sf
rule}  \geq T \right) I_{\{ \tau^{\sf rule}  \geq T \}}
  \stackrel{\mathrm{a.s.}}{=}
  \mathsf{E}^{({\bf d}(\ell_e))}_T \left( (\tau^{{\sf rule},({\cal N}_i)} - T)^+ |
  \mathcal{F}_{\{\tau^{\sf rule}  \geq T\}} \right)
                                                                       I_{\{
																		   \tau^{\sf
																		   rule}
																		   \geq
																		   T
																		   \}}$.
We conclude, from \ref{eqn:bound_derived_from_mei05}, that, as $c \to
\infty$,
 $ \mathsf{E}^{({\bf d}(\ell_e))}_T \left( (\tau^{{\sf rule},({\cal N}_i)} - T)^+ | \tau^{\sf
 rule}  \geq T \right) 
  \leq  \frac{c}{I} (1 + o(1))$.  


%



\ifCLASSOPTIONcaptionsoff
  \newpage
\fi



\bibliographystyle{IEEEtran}
\bibliography{IEEEabrv,premkumar-etal11detection-isolation-large-wsn}  
%
%








\end{document}